\documentclass[12pt,reqno]{amsart}
\usepackage[left=1.25in,right=1.25in,top=1.25in,bottom=1.25in]{geometry}
\usepackage[colorlinks=true,allcolors=Blue,linkcolor=Blue]{hyperref}
\usepackage{color}
\usepackage[dvipsnames]{xcolor}
\usepackage{amsmath, amsthm, amssymb, amsfonts}
\usepackage{mathtools}
\usepackage{parskip}
\usepackage{natbib}
\usepackage{url}
\usepackage{comment}
\usepackage{enumitem}
\usepackage[font=footnotesize]{caption}
\usepackage{graphicx}
\usepackage{tikz}
\usepackage{caption}
\usepackage{subcaption}
\usepackage{pgfplots}
\usepackage{float}
\usepgfplotslibrary{fillbetween}
\usetikzlibrary{patterns}
\pgfplotsset{compat=newest}
\usetikzlibrary{decorations.pathreplacing,angles,quotes}
\tikzset{
	declare function={
		normcdf(\x,\m,\s)=1/(1 + exp(-0.07056*((\x-\m)/\s)^3 - 1.5976*(\x-\m)/\s));
	}
}
\tikzstyle{dashed}=[dash pattern=on 8pt off 3pt]

\newtheorem{corollary}{Corollary}

\newtheorem{lemma}{Lemma}

\newtheorem{proposition}{Proposition}

\theoremstyle{definition}

\DeclareMathOperator\supp{supp}

\begin{document}
\setlength{\footnotesep}{0.75\baselineskip}
\setlength{\footskip}{1.25\normalbaselineskip}
	
\begin{titlepage}
	\title{Persuasion with Ambiguous Receiver Preferences}
	\author{\larger \textsc{Eitan Sapiro-Gheiler\\MIT\\ \\ \today}}
	\thanks{Email: eitans@mit.edu
		\\ \indent I thank Drew Fudenberg, Stephen Morris, Frank Schillbach, Dmitry Taubinsky, participants in MIT Theory Lunch, participants in MIT 14.192, and especially Alexander Wolitzky for helpful discussions and comments. I also thank two anonymous referees for their feedback and suggestions. This material is based upon work supported by the National Science Foundation Graduate Research Fellowship under Grant No. 1745302.}
	
	\begin{abstract}
		
		I describe a Bayesian persuasion problem where Receiver has a private type representing a cutoff for choosing Sender's preferred action, and Sender has maxmin preferences over all Receiver type distributions with known mean and bounds. This problem can be represented as a zero-sum game where Sender chooses a distribution of posterior mean beliefs that is a mean-preserving contraction of the prior over states, and an adversarial Nature chooses a Receiver type distribution with the known mean; the player with the higher realization from their chosen distribution wins. I formalize the connection between maxmin persuasion and similar games used to model political spending, all-pay auctions, and competitive persuasion. In both a standard binary-state setting and a new continuous-state setting, Sender optimally linearizes the prior distribution over states to create a distribution of posterior means that is uniform on a known interval with an atom at the lower bound of its support.
		
		\ \\
		
		\noindent\emph{JEL\ Classification:}\ D81, D82, D83 \newline
		
	\end{abstract}
	
	\clearpage
	
	\begingroup
	\let\MakeUppercase\relax
	\maketitle
	\thispagestyle{empty}
	\endgroup
\end{titlepage}

\newpage

\section{Introduction}
\label{sec:intro}
Consider a politician who is deciding how to disclose information about the cost-effectiveness of a new welfare program, but does not know how much spending voters will support. All voters have the same prior beliefs, but some will only approve if they expect the program to provide a high level of benefits per dollar spent, while others are willing to support even a moderately inefficient government outlay. Rather than imposing a prior distribution over preferences, the politician wishes to be robust to the worst-case distribution she may face given a known threshold for the average voter. In this setting, what disclosure rule maximizes the share of voters who approve of the welfare program after taking into account the politician's message? How do the optimal rule and the politician's utility differ from the case where the politician faces a known distribution of citizen preferences?

I address and generalize those questions through a model of Bayesian persuasion \citep{KG2011}, where a Sender commits to a message distribution in each state of the world and a Receiver uses Bayesian updating to form a posterior belief about the state based on the message structure. To represent Receiver's preferences, I use private types denoting the cutoff above which Receiver chooses Sender's preferred action. Sender knows the mean and support of Receiver types, and has maxmin preferences \citep{GS1989} over all Receiver type distributions satisfying those constraints. Regardless of the true state of the world, Sender maximizes the probability of inducing the favorable action. This model captures situations where all Receiver types process information in the same way, but may have different preferences over outcomes. In addition to the political spending example described above, a model of this style also applies to a variety of other situations, such as disclosing information about product quality (if potential customers share a prior belief about quality, but may be more or less picky about when they buy) or screening job candidates (if all firms have a common prior about candidate quality and see the same resum\'e, but have different thresholds for hiring).

This persuasion model can be reinterpreted as a zero-sum game between Sender and an adversarial Nature. Following the Bayesian persuasion literature, I can allow Sender to directly choose any distribution of posterior mean beliefs about the state that is a mean-preserving contraction of the prior. Then, Nature chooses a Receiver type distribution with the appropriate mean and domain; this choice is equivalent to choosing a mean-preserving contraction of a Receiver type distribution with support $\left\{0, 1\right\}$. The player with the higher realization from their chosen distribution wins the game. Such \textit{mean-preserving contraction games} (henceforth MPC games), albeit with simultaneous moves, have been studied in prior literature outside of the persuasion context (for example by \citealt{M1993}), as well as being used to represent competition between many Senders persuading a single Receiver (as in \citealt{BC2015}). Many of those works emphasize the role of uniform distributions, which induce indifference among many possible strategies for the opposing players. Adapting these results to my setting, I show that in a binary-state setting where the probability of the high state is weakly less than $1/2$, Sender's unique optimal posterior distribution places an atom at 0 and is uniform on an interval $[0, c]$ for $c \leq 1$. In doing so, I formalize the connection between maxmin persuasion and MPC games and show that the sequential timing of the maxmin persuasion game does not affect Sender's optimal distribution but the tie-breaking rule sometimes does. I also use a geometric approach based on the concavification argument of \citet{KG2011} to show that for any finite number of states of the world, or when the state is continuous and unimodal, a similar distribution\textemdash uniform on $[a, b] \subset (0, 1)$ with an atom at $a$\textemdash is one of many optimal distributions for Sender. The continuous-state setting is a novel specification of both the MPC game and the maxmin persuasion problem.

\section{Related Literature}
\label{sec:rl}
This work builds on the Bayesian persuasion problem of \citet{KG2011}, and adopts a similar approach to existing work in robust mechanism design. In addition, my model resembles a class of games I call MPC games, which include a continuous version of the Colonel Blotto game as well as competitive Bayesian persuasion by multiple Senders. I discuss the first two topics here and postpone discussion of the third to Section \ref{subsec:mpc}, after presenting the formal model.

In the baseline Bayesian persuasion model of \citet{KG2011}, Receiver has no private information. Subsequent literature in this area is surveyed in detail by \citet{K2019} and \citet{BM2019}, so I focus on the two works most directly related to the model I propose, \citet{KMZL2017} and \citet{HW2020}.\footnote{Other works use maxmin preferences in Bayesian persuasion settings, but are much more distinct. In \citet{K2020}, possible Receiver type distributions are distortions of a ``reference distribution;" in \citet{DP2020}, there is full ambiguity about Receiver's posterior belief; and in \citet{LR2017} and \citet{DLL2019}, Receiver has maxmin preferences.} The former has an interval state space, Receiver types that enter payoffs linearly, and a binary action, as in my model; however, it endows Sender with a prior distribution over Receiver types. If that prior distribution is log-concave, then the optimal distribution for Sender can be generated by upper censorship; the resulting distribution of posterior means is essentially a truncated version of the prior where states in some interval $[\alpha, 1]$ are replaced with an atom at $\beta \in (\alpha, 1)$. In the continuous-state version of my model, linearizing the prior rather than censoring high states helps Sender avoid facing a tailored Receiver type distribution in response. To make sure this strategy respects Bayes-plausibility, Sender may use a truncated uniform distribution with interior support.

The model of \citet{HW2020} is most similar to the one considered here: it is a binary-action model where Sender has maxmin preferences over Receiver types and maximizes the probability of inducing the favorable action. However, Receiver types represent an ambiguous posterior about a binary state of the world rather than a payoff-relevant characteristic which does not directly interact with beliefs about the state. This model captures substantively different applications\textemdash e.g., voters with common ideology who privately read outside news sources before listening to a politician's speech, rather than the equally-informed voters with different ideological positions in my model. Working with belief-independent Receiver types also means that I am able to characterize Receiver's posterior distribution and thus provide a sharp testable prediction\textemdash all posteriors in a known interior interval are equally likely. Methodologically, because my formulation features a simpler interaction between Receiver's type and Sender's signal, I am able to extend my approach to a continuous-state case.

A literature in robust mechanism design has also used moment conditions alongside maxmin preferences. \citet{W2016} considers a bilateral trade model where each agent has a valuation in $[0, 1]$ and knows only the mean of the other agent's type distribution. In that model, agents' worst-case beliefs have binary support. Here, it is similarly possible to define a binary-support worst-case Receiver type distribution, but Sender's desire to induce indifference between many such distributions means the optimal posterior distribution has interval support. In \citet{CLMH2019}, a principal with maxmin preferences offers a surplus-maximizing contract to a privately informed agent. Similar to my model, the agent's type distribution has known mean and support $[0, 1]$. As in \citet{HW2020} and my work, the optimal mechanism for the principal induces a payoff that is piecewise linear in the agent's type. Finally, \citet{C2018} considers a setting where a seller with maxmin preferences faces an unknown distribution of buyer valuations. The seller knows the first $N - 1$ moments of the valuation distribution and an upper bound on the $N$th moment. Similar to the concavification argument of \citet{KG2011}, optimal transfers are given by the non-negative monotonic hull of a degree-$N$ polynomial.

\section{Model}
\label{sec:model}
\subsection{Setup and Preferences}
\label{subsec:pref}
There is one Sender (she) and one Receiver (he).\footnote{The presence of one Receiver with an unknown type may also be interpreted as a population of Receivers, each with a known type, with which Sender communicates publicly.} Both players share a common prior $F \in \Delta([0, 1])$ about the state of the world $\omega \in [0, 1]$, with $\mathbb{E}_F[\omega] = \pi \in (0, 1)$. Only Receiver knows his private type $r \in [0, 1]$, but the mean Receiver type $r^* \in (0, 1)$ is common knowledge. In Section \ref{sec:ext}, I describe potential relaxations of these assumptions which endow Sender with less precise information about states or Receiver types. 

Sender considers potential Receiver type distributions $T$ in the set 
\begin{equation*}
	\mathcal{T} = \left\{\text{cdf T over } [0, 1] \, \bigg| \, \int r \, dT(r) = r^* \right\}.
\end{equation*}
I restrict Sender to the standard Bayesian persuasion tool of committing ex-ante to a Blackwell experiment, i.e., a state-dependent signal distribution, and in particular do not allow her to elicit Receiver's type in order to capture the public-communication interpretation of this model. After Sender communicates, Receiver chooses a binary action $a \in \left\{0, 1\right\}$ whose utility depends on the state and on Receiver's type:
\begin{equation*}
	u_R(a, \omega, r) = a \, (\omega - r).
\end{equation*}
Thus when Receiver believes $\mathbb{E}[\omega \, | \, \text{Sender's message}] > r$, he strictly prefers $a = 1$, and when the opposite inequality holds he strictly prefers $a = 0$.\footnote{Receiver's choice when indifferent will not affect equilibrium outcomes, but will affect Sender's equilibrium strategy. I discuss this tie-breaking issue in Section \ref{subsec:tb}.} The explicit functional form used here is for ease of exposition only. Whenever Receiver's utility is a linear function of the state, his action depends only on the mean of his posterior belief about the state, and my results still hold (under an appropriate re-normalization of the interval of Receiver types).

Sender's goal is to maximize the probability of inducing the high action $a = 1$ independent of the true state $\omega$ and true Receiver type $r$:
\begin{equation*}
	u_S(a, \omega, r) = a.
\end{equation*}

\subsection{The Maxmin Persuasion Problem}
\label{subsec:maxmin}
Since Receiver's choice of action depends only on the mean $q$ of the posterior belief distribution, I can follow \citet{B1953} and directly consider Sender choosing a distribution of posterior means $G$ such that $G$ is a mean-preserving contraction of the prior distribution $F$. The set of feasible distributions of posterior means is therefore
\begin{equation*}
	\begin{split}
		\mathcal{G} =  \bigg\{ \text{cdf } G \text { over } [0, 1] \, \bigg| & \int_0^x G(q) \, dq \leq \int_0^x F(q) \, dq \hspace{0.5em} \forall \hspace{0.5em} x \in [0, 1]
		\\& \text{and } \int_0^1 G(q) \, dq = \int_0^1 F(q) \, dq \bigg\}.
	\end{split}
\end{equation*}
I follow the literature in referring to this constraint as Bayes-plausibility. Note that when $\supp(F) = \left\{0, 1\right\}$, a case which I refer to as \textit{binary support}, any posterior distribution that satisfies the equality at $x = 1$ satisfies the inequality for all $x \in [0, 1)$.

Using this formulation and Receiver's preferences, I rewrite Sender's utility as
\begin{equation*}
	u_S(q, r) =  \textbf{1}(q > r),
\end{equation*}
where I assume that an indifferent Receiver chooses Sender's less-preferred action, $a = 0$. Sender's full optimization problem is therefore
\begin{equation}
	\label{eq:gen}
	\begin{split}
		\max_{G \in \mathcal{G}} & \left\{\min_{T \in \mathcal{T}} \int \int \textbf{1}(q > r) \, dG(q) \, dT(r) \right\}.
	\end{split}
\end{equation}
I state the optimization problem using a maximum and minimum, rather than supremum and infimum; the tie-breaking rule for indifferent Receivers ensures that the maximum and minimum are well-defined (see the proof of Lemma \ref{lm:saddle} in Appendix \hyperref[subsec:a1]{A1} for details). The main difference from standard Bayesian persuasion with private information is the presence of an endogenously-determined Receiver type distribution. 

\subsection{MPC Games}
\label{subsec:mpc}
I characterize the solution to the maxmin persuasion problem by reframing Sender's maxmin preferences as a zero-sum game, in which Sender designs a distribution of posterior means and then Nature adversarially designs a type distribution. More generally, my persuasion model can be viewed as a special case of a more general game which I call an \textit{MPC game}. In this game, players $1,...,N$ simultaneously\footnote{Simultaneous choice is a feature of most prior literature on games of this kind; I discuss in Section \ref{subsec:tb} why the switch from sequential to simultaneous moves does not affect the result.} choose distributions $G_1,...,G_N$ that are mean-preserving contractions of corresponding distributions $F_1,...,F_N$. A realization $x_i$ is drawn from each distribution $G_i$ to produce a vector of realizations $x = (x_1,...,x_N)$. A prize allocation rule $A(x): \mathbb{R}^N \rightarrow \mathbb{R}^N$ determines each player's payoff as a function of the realizations.\footnote{This setup may remind the reader of the literature on contests. In recent work, \citet{AT2023} also combines information design and contests, but focuses on a standard cost-of-effort setup for contest participants with a third party providing information about value profiles. In contrast, I represent an information design problem as a contest with Nature as a participant.} One simple prize allocation rule is to assign the player with the highest realization a payoff of 1 and all other players a payoff of 0; such a rule fits my model, where Sender gets a payoff of 1 if and only if the realized posterior exceeds the realized Receiver type. Various choices of $F_i$ have been paired with this prize allocation rule in prior literature. In particular, as noted in the previous section, if a cdf $F_i$ over a positive interval $[0, c]$ or over $\mathbb{R}_+$ has binary support, then any $G_i$ with the same domain and mean as $F_i$ is a mean-preserving contraction of $F_i$. MPC games where the $F_i$ have binary support and domain $\mathbb{R}_+$ have been used to describe campaign spending or distribution of revenues by politicians \citep{M1993, CG1998, SP2006}. Changing the domain to a finite interval $[0, c]$ has been used to model all-pay auctions with complete information \citep{BKV1996, H2015}. Further specifying the domain as $[0, 1]$ and potentially allowing $F_i$ to have non-binary support can represent competition between different Senders attempting to persuade a single Receiver \citep{BC2015, HKB2019, AK2020}. However, this work is the first to explicitly use the connection between MPC games and persuasion by a single Sender with maxmin preferences. 

\section{The Binary-State Setting}
\label{sec:bin}
In this section, I fully characterize Sender's optimal distribution when the prior $F$ has binary support, so that a distribution of posterior means is the same as a posterior distribution (I use the latter expression for simplicity). This case is equivalent to a 2-player MPC game where $F_1$ equals $F$, with domain $[0, 1]$, binary support, and mean $\pi$; $F_2$, which represents Nature's mean constraint, has domain $[0, \infty)$, binary support, and mean $r^*$.\footnote{This equivalence also holds when $F$ does not have binary support. However, I do not use it in characterizing Sender's optimal distribution with a continuous-support prior in Section \ref{sec:cont}.} The solution to the maxmin persuasion problem of Equation (\ref{eq:gen}), as well as Sender's optimal posterior distribution under slight variations of my model, follows from extending earlier results about MPC games. Proposition \ref{prop:bin} shows that when the prior $\pi$ is weakly less than $1/2$, Sender uniquely selects an upper-truncated uniform distribution with an atom at $q = 0$. Under favorable tie-breaking\textemdash where an indifferent Receiver chooses $a = 1$ rather than $a = 0$\textemdash Corollary \ref{cor:tbGood} in Appendix \hyperref[subsec:a4]{A4} shows that Sender may modify this solution by also placing an atom at $q = 1$.

\subsection{Maxmin Persuasion as an MPC Game}
\label{subsec:tb}
Equivalence of the maxmin persuasion problem in Equation (\ref{eq:gen}) and the MPC game specified above rests on two results. The first (Lemma \ref{lm:saddle} in Appendix \hyperref[subsec:a1]{A1}) is that Nash equilibrium strategies for Sender in the MPC game are equivalent to optimal posterior distributions in the sequential-move game implied by Sender's maxmin preferences. This result follows from a minimax theorem in \citet{F1953}, which shows that Sender's maxmin and minmax utilities are equal, and therefore equal to the utility from the MPC game. 

The second result (Lemma \ref{lm:tb} in Appendix \hyperref[subsec:a1]{A1}) is that tie-breaking against Sender is equivalent to ignoring tie-breaking but allowing the Receiver type distribution to be unbounded above.\footnote{Similar observations have been made in the context of all-pay auctions by \citet{S2015} and \citet{GKR2022}. The latter model finds that when ties may occur on intervals with positive measure, players' equilibrium strategies involve multiple disjoint intervals with an atom at 0, rather than the single interval and atom at 0 that arises in my model.} Since Nature moves second, I use tie-breaking against Sender to ensure that the minimizing Receiver type distribution for each posterior distribution is well-defined; however, most results on MPC games use even tie-breaking, where a posterior $q = r$ convinces that Receiver type with probability $1/2$. With unfavorable tie-breaking, in order to persuade Receiver type $r$, Sender must generate posterior $q_r^\varepsilon = r + \varepsilon$ for arbitrary $\varepsilon > 0$. Thus, unlike in an MPC game with even tie-breaking, Sender can never persuade type $r = 1$. I can restore the equivalence between maxmin persuasion and MPC games by allowing Nature in the MPC game to generate a Receiver type $r = 1 + \eta$ for arbitrary $\eta > 0$. Then Sender can attain her even tie-breaking utility for all interior Receiver types in the $\varepsilon \rightarrow 0$ limit and Nature can replace any instance of $r = 1$ with $r = 1 + \eta$ in the $\eta \rightarrow 0$ limit.

\subsection{Optimal Posterior Distributions}
\label{subsec:binOpt}
Having established equivalence between the maxmin persuasion problem and an appropriate MPC game, the solution to the maxmin persuasion problem closely resembles Theorem 4 of \citet{H2015}. I extend that result by providing an alternative proof which shows uniqueness of Sender's optimal distribution when $\pi \leq 1/2$, as well as necessary and sufficient conditions on any maxmin-optimal posterior distribution when $\pi > 1/2$. Let $\delta_x$ be the Dirac distribution with all mass at $q = x$ and $U[x, y]$ be the uniform distribution over the interval $[x, y]$. The following describes Sender's optimal posterior distribution:
\begin{proposition}
	\label{prop:bin}
	Let $\supp(F) = \left\{0, 1\right\}$ and let ties be broken against Sender.\\
	If $\pi > 1/2$, then a posterior distribution $G^*$ is optimal for Sender if and only if $\mathbb{E}_{G^*}[\omega] = \pi$ and $G^*(x) \leq x \hspace{0.5em} \forall \hspace{0.5em} x \in [0, 1]$.\\ 
	If $\pi \leq 1/2$, then Sender's unique optimal posterior distribution $G^*$ is as follows:
	\begin{itemize}
		\item If $r^* \leq \pi \leq 1/2$,
		\begin{equation*}
			G^* = U[0, 2\pi].
		\end{equation*}
		\item If $\pi \leq r^* \leq 1/2$,
		\begin{equation*}
			G^* = \bigg(1 - \frac{\pi}{r^*}\bigg) \, \delta_0 + \frac{\pi}{r^*} \, U[0, 2r^*].
		\end{equation*}
		\item If $\pi \leq 1/2 \leq r^*$,
		\begin{equation*}
			G^* = (1 - 2\pi) \, \delta_0 + 2\pi \, U[0, 1].
		\end{equation*}
	\end{itemize}
\end{proposition}
\begin{proof}
	See Appendix \hyperref[subsec:a2]{A2}.
\end{proof}
\begin{figure}
	\begin{tikzpicture}[scale=1]
		\LARGE
		
		\begin{axis}
			[axis x line = middle,
			axis y line = left,
			xmin = 0, xmax = 1.2,
			ymin = 0, ymax = 1.2,
			ytick=1,
			xtick={0.00001,1},
			xticklabels={0, 1},
			clip=false]
			
			\draw [solid, line width = 1mm, orange] (axis cs: 0, 0) -- (axis cs: 2/3, 1);
			\draw [solid, line width = 1.15mm, orange] (axis cs: 2/3, 1) -- (axis cs: 1, 1);
			
			\draw [solid, line width = 1mm, blue] (axis cs: 0, 1/6) -- (axis cs: 4/5, 1);
			\draw [solid, line width = 0.85mm, blue] (axis cs: 4/5, 1) -- (axis cs: 1, 1);
			\node[circle,fill=blue,inner sep=0pt,minimum size=3mm] at (0, 1/6) {};
			
			\draw [solid, line width = 1mm, black!60!green] (axis cs: 0, 1/3) -- (axis cs: 1, 1);
			\node[circle,fill=black!60!green,inner sep=0pt,minimum size=4.24mm] at (0, 1/3) {};
		\end{axis}
	\end{tikzpicture}
	\vspace{-0.4cm}
	\caption{Optimal posterior distributions in the binary-state setting for different values of the mean state $\pi$ and mean Receiver type $r^*$.\\
		In blue, $r^* = 1/4 < 1/3 = \pi$, and $G^* = U[0, 2/3]$.\\
		In orange, $\pi = 1/3 < 2/5 = r^*$ and $G^* = (1/6) \, \delta_0 + (5/6) \, U[0, 4/5]$.\\
		In green, $\pi = 1/3 < 3/5 = r^*$ and $G^* = (1/3) \, \delta_0 + (2/3)\, U[0, 1]$.}
	\label{fig:binary}
\end{figure}
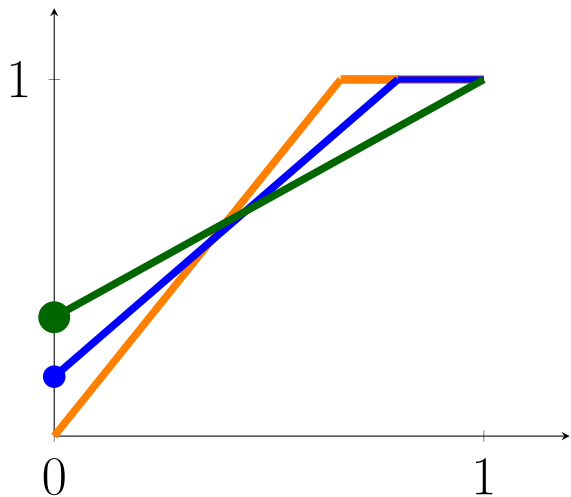
Because Nature can choose a binary support distribution where Receiver type $r = 0$ is always persuaded and Receiver type $r = 1$ is never persuaded, Sender's payoff cannot exceed $1 - r^*$. Whenever Sender chooses a posterior distribution with a convex cdf, this Receiver type distribution is indeed optimal for Nature, and Sender attains her maximum payoff. Since the uniform distribution $U[0, 1]$ has the smallest mean among distributions with convex cdfs, this choice is feasible for Sender if and only if $\pi \geq 1/2$ (and multiple such distributions are feasible when $\pi > 1/2$). When $\pi < 1/2$, Sender chooses a distribution that is as close to uniform as possible given her Bayes-plausibility constraint. This choice requires her to place an atom at $q = 0$, truncate the upper bound of the distribution's support below $q = 1$, or both. Fixing $\pi \leq 1/2$, for small $r^*$ Sender truncates the support at $2\pi$ but places no atom at 0. As $r^*$ increases, Sender simultaneously increases the size of the atom and moves the upper bound of the support towards 1. A higher average Receiver type makes high posteriors more valuable to Sender, so she is willing to sometimes fully reveal the low state in order to generate more of these posteriors. Figure \ref{fig:binary} shows three examples of optimal posterior distributions, corresponding to the three cases of Proposition 1.

\citet{H2015} does not establish uniqueness of this Nash equilibrium of the MPC game. A related work, \citet{NA2018}, shows through explicit calculations of players' utilities under various distributions that Sender's Nash equilibrium strategy is unique when $\pi \leq 1/2$ (Theorems 4 and 5 in that work) but gives only a partial characterization of optimal strategies for Sender when $\pi > 1/2$ (Theorem 10 in that work). In the maxmin persuasion setting with tie-breaking against Sender, I am able to avoid issues with limits of $\varepsilon$-approximating distributions and close that gap: Lemma \ref{lm:bigP} in Appendix \hyperref[subsec:a2]{A2} gives a necessary and sufficient condition for Sender's optimal distribution when $\pi > 1/2$. Additionally, in Lemma \ref{lm:bin} of Appendix \hyperref[subsec:a3]{A3}, I provide a novel geometric proof of Sender's optimal posterior distribution for the case $\pi \leq 1/2$, including its uniqueness, which leverages the concavification approach of \citet{KG2011}. This proof informs my approach in the continuous-state setting.

\section{The Continous-State Setting}
\label{sec:cont}
While the solution when $F$ has binary support is especially sharp, that restriction may not always be plausible. In this section, I consider the maxmin persuasion problem of Equation (\ref{eq:gen}) when $F$ is a continuously differentiable and unimodal cdf over $[0, 1]$ with $F(0) = 0$. I assume that, for some mode $m \in [0, 1]$, the density $f$ is strictly increasing on $[0, m)$ and strictly decreasing on $(m, 1]$. To rule out the binary-state solution, I also assume that $f'(0) < 1 - 2\pi$. In this setting, a double-truncated uniform distribution of posterior means is optimal when $r^*$ is sufficiently small (Proposition \ref{prop:contSmall}) or large (Proposition \ref{prop:contLarge}). Analogous results hold when, rather than being continuous, $F$ is supported on $N$ values $\left\{q_1,...,q_N\right\} \in [0, 1]^N$ with $N > 2$; I provide full details and proofs of this extension in Appendix \hyperref[subsec:b6]{B6}. Before turning to the main result, I first discuss two simple cases which extend the intuitions of the binary-state setting.

\subsection{Simple Continuous Priors}
\label{subsec:simpCont}
In the continuous-state setting, Sender's constraint is different from Nature's. It is no longer true that any cdf $G$ with support $[0, 1]$ and mean $\pi$ is a mean-preserving contraction of the prior $F$. Instead, the chosen cdf must additionally satisfy the \textit{integral constraint}
\begin{equation*}
	\int_0^x G(q) \, dq \leq \int_0^x F(q) \, dq \hspace{0.5em} \forall \hspace{0.5em} x \in [0, 1].
\end{equation*}
Under the assumptions $F(0) = 0$ and $f'(0) < 1 - 2\pi$, this constraint prevents Sender from choosing any of the optimal distributions in Proposition \ref{prop:bin}, as they violate it in the interval $(0, \varepsilon)$ for $\varepsilon > 0$ sufficiently small.

Despite this new constraint, two cases of the continuous-state model are easy to solve using the intuitions of the previous section. Nature may still choose the binary support distribution which generates only Receiver types $r = 0$ and $r = 1$, so Sender's utility is still upper-bounded by $1 - r^*$. Thus for any $F$ that first-order stochastically dominates $U[0, 1]$\textemdash so that $F(q) \leq q \hspace{0.5em} \forall \hspace{0.5em} q \in [0, 1]$\textemdash it is easy to see that full disclosure is optimal, since it ensures that $G^* = F$ and Sender's utility attains the upper bound. This condition generalizes the case where $F$ is unimodal with $m = 1$.

When $F$ is concave but not uniform, it must be that $F$ lies strictly above $U[0, 1]$ on $(0, 1)$ and therefore that $\pi < 1/2$. Additionally, the uniform distribution $U[0, 2\pi]$, which was uniquely optimal when $r^* \leq \pi \leq 1/2$ in the binary-state setting of Proposition \ref{prop:bin}, satisfies the integral constraint. To see why, note that the shape of $F$ ensures that $U[0, 2\pi]$ lies strictly below $F$ on $(0, c)$ for some $c < 1$; therefore the integral constraint is satisfied with equality at $x = 0$ and strict inequality for $x \in (0, c]$. The difference between the left- and right-hand sides of the constraint strictly decreases for $x \in (c, 1)$, but only reaches 0 at $x = 1$: thus the weak inequality is preserved on the entire interval $[0, 1]$.\footnote{This geometric approach to the integral constraint will be key in proving Proposition \ref{prop:contSmall}.} The binary-state maxmin persuasion problem is a relaxed version of the continuous-state maxmin persuasion problem, so a feasible optimal solution for Sender in the former must be optimal in the latter. Thus, if $r^* \leq \pi$ and $F$ is concave but not uniform, then $G^* = U[0, 2\pi]$ is the unique optimal distribution for Sender. This condition generalizes the case of unimodal $F$ with $m = 0$.

In the maxmin setting, Nature's mean constraint represents information Sender possesses which allows her to consider only a particular set of possible Receiver type distributions. Thus it is reasonable for Nature to face only a mean constraint while Sender also faces the integral constraint. However, in the case of competitive persuasion, where both parties are persuading Receiver about a common state, it is natural to require all players to choose mean-preserving contractions of the same continuous prior. This case is studied in \citet{HKB2019}, where the optimal distribution of posterior means divides the prior support into finitely many intervals and alternates between matching the prior and generating a linear mean-preserving contraction on each interval. Nature's weaker constraint in my setting rules out this result.

\subsection{Optimal Distributions with Small $r^*$}
\label{subsec:smallR}
Towards providing sufficient conditions for a double-truncated uniform distribution (henceforth DTU) to be optimal for Sender, I first establish notation. A DTU places no mass on any $q \in [0, \ell]$, an atom at $q = \ell$, uniform mass on all $q \in [\ell, u]$, and no mass on any $q \in [u, 1]$. It thus has three parameters: the lower truncation length, the size of the atom at $\ell$, and the upper truncation length. Because any Bayes-plausible DTU's mean must be $\pi$, the atom size is uniquely determined by the truncation lengths. That is, given $[\ell, u] \subseteq [0, 1]$ there is only one DTU with mean $\pi$ and support $[\ell, u]$. I can thus characterize a DTU by the slope $\beta > 0$ and intercept $y \in [0, 1)$ of the uniform portion of its cdf, writing it as $G_y^\beta$. Explicit formulas for the relationship between truncation lengths, atom size, and slope, as well as bounds on these parameters, are in Appendix \hyperref[subsec:b1]{B1}. For each $y$, there is a slope $\beta(y)$ which delivers Sender's highest utility among DTUs with intercept $y$ (Lemma \ref{lm:intOpt} in Appendix \hyperref[subsec:b2]{B2}); I refer to the DTU $G_y^{\beta(y)}$ as $y$-optimal. Let $q_i(\beta, 0)$ be the smallest nonzero point of intersection between the DTU $G_y^\beta$ and the prior $F$. Figure \ref{fig:dtu} shows an example of the $0$-optimal DTU when $F$ is a truncated normal distribution, highlighting the notation above. With notation fixed, the following proposition describes Sender's choice for small $r^*$: 
\begin{figure}[t!]
	\begin{tikzpicture}[scale=1]
		\LARGE
		
		\begin{axis}
			[axis x line = middle,
			axis y line = left,
			xmin = 0, xmax = 1.2,
			ymin = 0, ymax = 1.2,
			ytick={0.000001,1},
			yticklabels={y,1},
			xtick={0.000001,0.18375, 0.3248, 0.466, 1},
			xticklabels={0,$\ell$,,, 1},
			clip=false]
			
			\addplot [domain=0:1, line width = 0.75mm, smooth, blue]				{(normcdf(3*(x-1/5),1/5,1/3) - 0.0082)/(1-0.0082)};
			
			\addplot [domain=-0:0.18375, line width = 0.75mm, smooth, orange]	{0};
			\addplot [domain=0.18375:0.4675, line width = 0.75mm, smooth, orange]	{2.14665*x};
			\addplot [domain=0.4655:1, line width = 0.75mm, smooth, orange]	{1};

			\draw [dashed, line width = 0.75mm, orange] (axis cs: 0.18375, 0) -- (axis cs: 0.18375, 2.14665*0.18375);
			
			\draw [solid, line width = 0.1mm, black] (axis cs: 0.3248, 0) -- (axis cs: 0.3248, 1.2);
			
			\draw [solid, line width = 0.1mm, black] (axis cs: 0.18375, 0) -- (axis cs: 0.18375, 1.2);
			
			\draw [solid, line width = 0.1mm, black] (axis cs: 0.466, 0) -- (axis cs: 0.466, 1.2);
			
			\node[label={$q$}] 		at (axis cs: 1.2, -0.02) 	
			{};
			
			\node[label={$q_i$}] 		at (axis cs: 0.3248, -0.285) 	
			{};
			
			\node[label={$u$}] 		at (axis cs: 0.466, -0.265) 	
			{};
		\end{axis}
	\end{tikzpicture}	
	\vspace{-0.8cm}
	\caption{An example $0$-optimal DTU, $G^{2.15}_0$, in orange. The prior $F$ (blue) is a $N(1/5, 1/3)$ distribution truncated in $[0, 1]$. Its support is $[\ell, u] = [0.18, 0.47]$; its intercept is $y = 0$; its slope is $\beta(0) = 2.15$; the smallest nonzero point of intersection between the DTU and the prior is $q_i(\beta(0), 0) = 0.32$.}
	\label{fig:dtu}		
\end{figure}
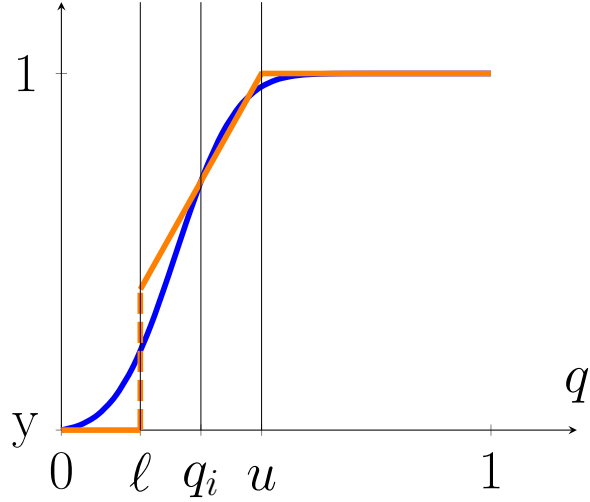

\begin{proposition}
	\label{prop:contSmall}
	Let $r^* \in (0, q_i(\beta(0), 0)]$. The $0$-optimal double-truncated uniform distribution, $G_0^{\beta(0)}$, attains Sender's highest possible utility. Any other distribution which attains that utility must have the same concavification as $G_0^{\beta(0)}$.
\end{proposition}
\begin{proof}
	See Appendix \hyperref[subsec:b5]{B5}.
\end{proof}

The key step in the proof is to show that the integral constraint binds only at a single interior point, the intersection $q_i(\beta(0), 0)$ between the DTU and the prior cdf (Lemma \ref{lm:contEll} in Appendix \hyperref[subsec:b3]{B3}).\footnote{When $F$ has finite support, the integral constraint need not bind at any interior points. In Appendix \hyperref[subsec:b6]{B6}, I show that if it does bind, it must do so only on a subset of $\supp(F)$. Letting $q_\text{min}$ be the minimal element of $\supp(F)$ where the integral constraint binds, Proposition 2 holds after replacing $q_i(\beta(0), 0)$ with $q_\text{min}$.} Using this result, I can adapt the the strategy used in my geometric proof of Proposition \ref{prop:bin}. Towards simplifying the integral constraint, note that a DTU $G_0^\beta$ will have zero, one, or two interior intersections with $F$ depending on its slope.\footnote{Slight variations of these cases may occur and are dealt with in the proofs of Appendix B, but can be ignored to simplify the intuition.} Writing the integral constraint as a function of $x$,
\begin{equation*}
	v(x) = \int_0^x F(q) \, dq - \int_0^x G_0^\beta (q) \, dq,
\end{equation*}
the intersections of $F$ and $G_0^\beta$ can be used to infer whether $v$ is increasing or decreasing on particular intervals. Figure \ref{fig:lowerB} shows an example of this approach. Combined with the observation that $v(0) = v(1) = 0$, this behavior allows me to show that if $G_0^\beta$ has two interior intersections with $F$, then it satisfies the integral constraint if and only if $v(q_i(\beta, 0)) \geq 0$. To select among these Bayes-plausible DTUs, note that when $y = 0$ the slope of a concavified DTU $G_0^\beta$ equals $\beta$ in the lower truncation interval $[0, \ell]$ and uniform interval $[\ell, u]$. Thus Sender's $0$-optimal DTU is given by making $\beta$ small (to minimize Nature's utility) while satisfying the simplified integral constraint.
\begin{figure}[t!]
	\centering
	\begin{tikzpicture}[scale=0.65]
		\LARGE
		
		\begin{axis}
			[axis x line = middle,
			axis y line = left,
			xmin = -3, xmax = 4,
			ymin = 0, ymax = 1.2,
			ytick=1,
			xtick={-2.9999,3},
			xticklabels={0,1},
			clip=false]
			
			\addplot [domain=-3:3, line width = 0.75mm, smooth, blue]				{normcdf(x,-1,0.85)-0.005};
			
			\addplot [domain=-3:-1.5, line width = 0.75mm, smooth, orange]	{0};
			\addplot [domain=-1.5:1/3+0.01, line width = 0.75mm, smooth, orange]	{0.3*(x+3)};
			\addplot [domain=1/3:3, line width = 0.75mm, smooth, orange]	{1};
			
			\draw [dashed, line width = 0.75mm, orange] (axis cs: -1.5, 0) -- (axis cs: -1.5, 0.3*1.5);
			
			\draw [solid, black] (axis cs: -1.5, -0.37) -- (axis cs: -1.5, 1.1);
			
			\node[label={\normalsize $0 \nearrow + $}] 	at (axis cs: -2.25, -0.4) 	{};
			\node[label={\normalsize $+ \searrow 0$}] 	at (axis cs: 0.75, -0.4) 	{};
			
			\node[label={$q$}] 		at (axis cs: 3.5, -0.02) 		{};
		\end{axis}
	\end{tikzpicture}	
	\hspace{1cm}
	\begin{tikzpicture}[scale=0.65]
		\LARGE
		
		\begin{axis}
			[axis x line = middle,
			axis y line = left,
			xmin = -3, xmax = 4,
			ymin = 0, ymax = 1.2,
			ytick=1,
			xtick={-2.9999,3},
			xticklabels={0,1},
			clip=false]
			
			\addplot [domain=-3:3, line width = 0.75mm, smooth, blue]				{normcdf(x,-1,0.85)-0.005};
			
			\addplot [domain=-3:-2, line width = 0.75mm, smooth, orange]	{0};
			\addplot [domain=-2:1.01, line width = 0.75mm, smooth, orange]	{0.25*(x+3)};
			\addplot [domain=1:3, line width = 0.75mm, smooth, orange]	{1};

			\draw [dashed, line width = 0.75mm, orange] (axis cs: -2, 0) -- (axis cs: -2, 0.25);
			
			\draw [solid, black] (axis cs: -2, -0.37) -- (axis cs: -2, 1.1);
			\draw [solid, black] (axis cs: -0.95, -0.37) -- (axis cs: -0.95, 1.1);
			\draw [solid, black] (axis cs: 0.92, -0.37) -- (axis cs: 0.92, 1.1);
			
			\node[label={\small $0 \nearrow +$}] 	at (axis cs:-2.6, -0.4) 	{};
			\node[label={\small $+ \searrow \, ?$}] 	at (axis cs:-1.5, -0.4) 	{};
			\node[label={\small $? \nearrow \, ?$}] 	at (axis cs: 0, -0.4) 		{};
			\node[label={\small $? \searrow 0$}] 	at (axis cs: 2, -0.4) 		{};
			
			\node[label={$q$}] 		at (axis cs: 3.5, -0.02) 		{};
		\end{axis}
	\end{tikzpicture}	
	\vspace{-0.4cm}
	\caption{Two DTUs with intercept 0 and varying slopes (orange) and their relationship to the prior $F$ (blue). The guidelines show where the function $v$, which represents the difference between the integral of $F$ and that of the corresponding DTU, switches from increasing to decreasing or vice-versa.}
	\label{fig:lowerB}		
\end{figure}
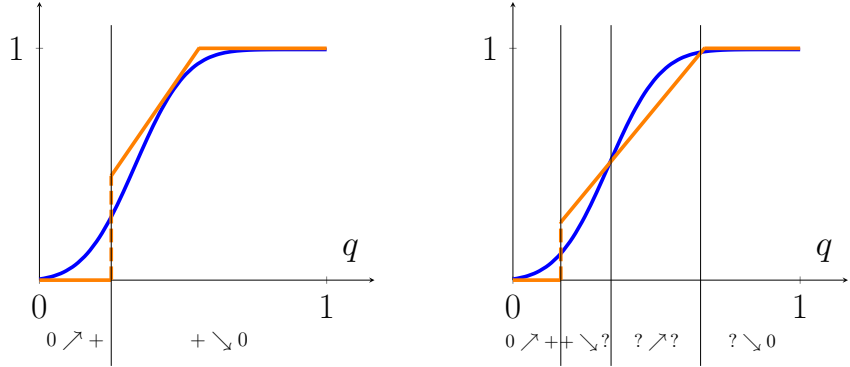

The remainder of the proof uses the fact that the integral constraint binds only at $q \in \left\{0, q_i(\beta(0), 0), 1\right\}$ for $G_0^{\beta(0)}$ and adapts the strategy used in my geometric proof of Proposition \ref{prop:bin}. First, I show that any DTU which delivers higher utility for Sender than $G_0^{\beta(0)}$ cannot be Bayes-plausible. I then approximate an arbitrary optimal distribution $H$ by a DTU, using the concavification of $H$ to ensure that this DTU upper-bounds $H$ everywhere above the lower truncation. If $H$ delivers Sender strictly higher utility than $G_0^{\beta(0)}$, the approximating DTU must do so as well; therefore it cannot be Bayes-plausible, and neither is $H$ itself. If $H$ delivers Sender the same utility as $G_0^{\beta(0)}$, then the approximating DTU is precisely $G_0^{\beta(0)}$ and the concavification of $H$ equals that of $G_0^{\beta(0)}$. This last step relies crucially on the slope of $G_0^{\beta(0)}$ being the same as that of its concavification in the lower truncation interval $[0, \ell]$. For a DTU with intercept $y > 0$, this property will no longer hold, and as a result the concavified optimal distribution will no longer be unique.

While the concavification of $G_0^{\beta(0)}$ is the unique concavified distribution that maximizes Sender's utility, $G_0^{\beta(0)}$ is not itself a unique solution to the maxmin persuasion problem. In the binary-state setting, Sender's optimal distribution was equal to its concavification everywhere on $[0, 1]$, and any other distribution with the same concavification would have a different mean. In the continuous-state setting, a DTU differs from its concavification on the lower truncation interval, so it is possible for a non-DTU distribution $H$ to have the same mean and concavification as a DTU. Thus uniqueness of the concavification is the strongest result that can be obtained.

\subsection{Optimal Distributions with Large $r^*$}
\label{subsec:bigR}
When $r^* > q_i(\beta(0), 0)$, characterizing both the optimal DTU and optimal distributions more generally becomes more difficult. In fact, an optimal distribution may not exist, though Sender's supremum utility over a sequence of distributions converging to optimality is always well-defined. Despite these challenges, I can still show that for sufficiently large values of $r^*$, DTUs are not dominated by other distributions. This result holds without alteration when $F$ has finite support:
\begin{proposition}
	\label{prop:contLarge}
	Let $r^* \in [\pi, 1)$. Then no distribution of posterior means gives Sender strictly higher utility than all DTUs.
\end{proposition}
\begin{proof}
	See Appendix \hyperref[subsec:b5]{B5}, or Appendix \hyperref[subsec:b6]{B6} for the finite-state case.
\end{proof}
This proof is similar in approach to that of Proposition \ref{prop:contSmall}. When $y > 0$, the slope of a concavified DTU $G_y^\beta$ is larger in the lower truncation interval $[0, \ell]$ than that of the DTU itself (since the concavification passes through the origin, while the DTU has intercept $y$), but is again equal to $\beta$ in the uniform interval $[\ell, u]$. Thus the concavified DTU has a kink at $q = \ell$. However, setting $r^* \geq \pi$ ensures that the kink does not affect the value of the concavified DTU at $r^*$. Then, as in Proposition \ref{prop:contSmall}, Sender's $y$-optimal DTU for each intercept $y$ is given by minimizing the slope $\beta$ subject to the integral constraint. Unlike in that proposition, I cannot directly characterize which choice of intercept is optimal. In fact, since the set $[0, 1)$ of possible intercept choices is not compact, it may be that the optimal choice is $y = 1$ and Sender's highest utility is attained only in the limit. However, I can still use the fact that each $y$-optimal DTU has minimal slope among Bayes-plausible DTUs with intercept $y$ to extend the bounding argument of Proposition \ref{prop:contSmall}. This approach rules out as infeasible any distribution that delivers strictly higher utility than all DTUs, but again leaves room for non-DTU distributions that attain Sender's highest possible utility.

\subsection{Non-Uniform Optimal Distributions}
\label{subsec:nonDTU}
The result of Proposition \ref{prop:contSmall} provides an appealing reason for focusing on DTUs as opposed to other maxmin-optimal posterior distributions: outside of the lower-truncation region, the optimal DTU is precisely equal to the unique optimal concavification. However, in the large-$r^*$ case of Proposition \ref{prop:contLarge}, the optimal concavification is no longer unique. To see why, assume an optimal DTU with $y > 0$ exists. Its concavification passes through the origin rather than the point $(0, y)$, so it has a kink at $q = \ell$. This kink can be used to alter the DTU without affecting Sender's utility. In particular, consider a distribution that places slightly positive mass in the interval $[\ell - \varepsilon, \ell)$, has a smaller atom than the DTU at $q = \ell$, and places slightly less mass than the DTU in the interval $(\ell, \ell + \varepsilon]$. This distribution, shown in Figure \ref{fig:2kink}, changes slope at $\ell$ and $\ell + \varepsilon$, but is equal to the DTU for $q \notin (\ell - \varepsilon, \, \ell + \varepsilon)$. Whenever $r^* \geq \ell + \varepsilon$, and in particular when $r^* \geq \pi$ (the case in Proposition \ref{prop:contLarge}) the deviation delivers the same utility for Sender.

\begin{figure}[t!]
	\centering
	\begin{tikzpicture}[scale = 0.95]
		\definecolor{dark-green}{RGB}{0, 100, 0}
		
		\pgfmathsetmacro{\p}{1/2}
		\pgfmathsetmacro{\l}{1/3}
		\pgfmathsetmacro{\k}{1/5}
		\pgfmathsetmacro{\b}{((\p - \k*\l)-((\p - \k*\l)^2 - ((\l)^2)*((1 - \k)^2))^(1/2))/(\l^2)}
		
		\begin{axis}[
			legend cell align={left}, 
			legend style={draw=none},
			legend pos = {north west},
			axis x line = middle,
			axis y line = left,
			xmin = 0, xmax = 1.2,
			ymin = 0, ymax = 1.2,
			xlabel=$q$,
			ytick={1},
			xtick={0.000001,1},
			xticklabels={0,1},
			clip=false]
			
			\addlegendimage{very thick, smooth, blue}
			\addlegendentry{DTU}
			
			\addlegendimage{smooth, white}
			\addlegendentry{}
			
			\addlegendimage{very thick, smooth, dark-green}
			\addlegendentry{deviation}
			
			\addplot [domain=0:\l, very thick, smooth, blue]	{0};
			\addplot [domain=\l-0.0001:(1-\k)/\b, very thick, smooth, blue]	{\k + x*\b};
			\addplot [domain=(1-\k)/\b-0.0001:1, very thick, smooth, blue]	{1};
			
			\addplot [domain=0:\l, thick, dashed, blue]	{x*(\b*\l+\k)/\l};
			
			\addplot [domain=0:\l-0.05, very thick, smooth, dark-green]	{0};
			\addplot [domain=\l-0.05:\l, very thick, smooth, dark-green]	{0.25*(x - \l + 0.05)};
			\addplot [domain=\l+0.05:(1-\k)/\b, very thick, smooth, dark-green]	{\k + x*\b};
			\addplot [domain=(1-\k)/\b-0.0001:1, very thick, smooth, dark-green]	{1};
			
			\addplot [domain=0:\l+0.0005, thick, dashed, dark-green]	{1.3834*x};
			\addplot [domain=\l:\l+0.0505, very thick, solid, dark-green]	{1.1*x+0.0945};
			
			\draw [solid, black] (\l, 0) -- (\l, 1);
			\draw [solid, black] (\l+0.05, 0) -- (\l+0.05, 1);
		\end{axis}
	\end{tikzpicture}
	\hfill
	\begin{tikzpicture}[scale = 0.95]
		\definecolor{dark-green}{RGB}{0, 100, 0}
		
		\pgfmathsetmacro{\p}{1/2}
		\pgfmathsetmacro{\l}{1/3}
		\pgfmathsetmacro{\k}{1/5}
		\pgfmathsetmacro{\b}{((\p - \k*\l)-((\p - \k*\l)^2 - ((\l)^2)*((1 - \k)^2))^(1/2))/(\l^2)}
		
		\begin{axis}[
			legend cell align={left}, 
			legend style={draw=none},
			legend pos = {north west},
			axis x line = middle,
			axis y line = left,
			xmin = 0.3, xmax = \l+0.1,
			ymin = 0.4, ymax = 0.55,
			xlabel=$q$,
			ytick={0.4, 0.5},
			xtick={0.31,0.4},
			clip=false]
			
			\addlegendimage{very thick, smooth, blue}
			\addlegendentry{DTU}
			
			\addlegendimage{smooth, white}
			\addlegendentry{}
			
			\addlegendimage{very thick, smooth, dark-green}
			\addlegendentry{deviation}
			
			\addplot [domain=\l-0.0001:0.4, very thick, smooth, blue]	{\k + x*\b};
			
			\addplot [domain=0.3:\l, very thick, dashed, blue]	{x*(\b*\l+\k)/\l};
			
			\addplot [domain=\l+0.05:0.4, very thick, smooth, dark-green]	{\k + x*\b};
			
			\addplot [domain=0.3:\l+0.0005, very thick, dashed, dark-green]	{1.3834*x};
			\addplot [domain=\l:\l+0.0505, very thick, solid, dark-green]	{1.1*x+0.0945};
			
			\draw [solid, black] (\l, 0.4) -- (\l, 0.55);
			\draw [solid, black] (\l+0.05, 0.4) -- (\l+0.05, 0.55);
			
		\end{axis}
	\end{tikzpicture}
	\vspace{-0.35cm}
	\caption{A potential deviation (solid green) from a DTU (solid blue); the concavification of each distribution is shown by dashed lines of the same color. The second panel focuses on the concavified deviation's double kink (at both guidelines), while the concavified DTU's has only one kink (at the first guideline).}
	\label{fig:2kink}
\end{figure}
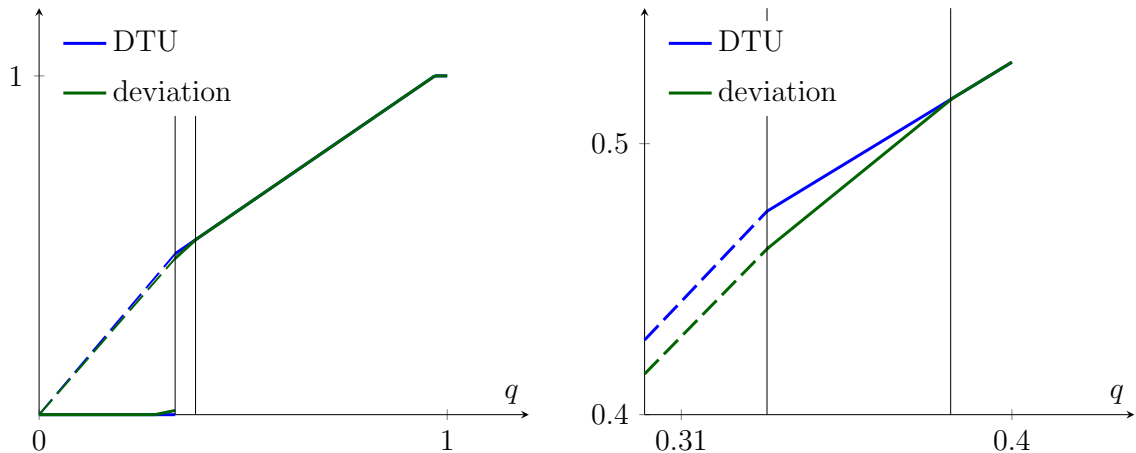

\subsection{Intermediate $r^*$}
The deviation in Figure \ref{fig:2kink} also sheds light on the difficulty of characterizing the optimal distribution when $r^* \in (q_1(\beta(0), 0), \pi)$. In the binary-state setting the optimal distribution equals its concavification. With a continuous state, lower truncation is one possible response to the integral constraint, but it is not a unique solution for Sender because that constraint may only bind at a finite set of interior points. For instance, the deviation in Figure \ref{fig:2kink} gives Sender a greater utility than the corresponding DTU when $r^* \in (0, \ell + \varepsilon)$ and is feasible whenever the integral constraint does not bind in that interval. Without further structure on the space of possible deviations from DTUs, even numerical approaches with a parametric prior distribution provide no insight, since they would require a novel algorithm to search over all mean-preserving contractions of the prior.

Despite this challenge, I am able to shed light on the prevalence of the intermediate-$r^*$ case by numerically estimating $q_i(\beta(0), 0)$ within a class of parametric prior distributions. For truncated normal priors\textemdash generated by taking a $N(\mu, \sigma^2)$ distribution and truncating it to lie in the unit interval\textemdash I show numerically that there is a gap between Propositions \ref{prop:contSmall} and \ref{prop:contLarge} only when $\mu < 1/2$ and $\sigma$ is large enough. For example, when $\mu = 0.2$, shown in orange in Figure \ref{fig:sims}, there is a gap only when $\sigma \geq 0.135$.
\begin{figure}
	\includegraphics[width=0.5\hsize]{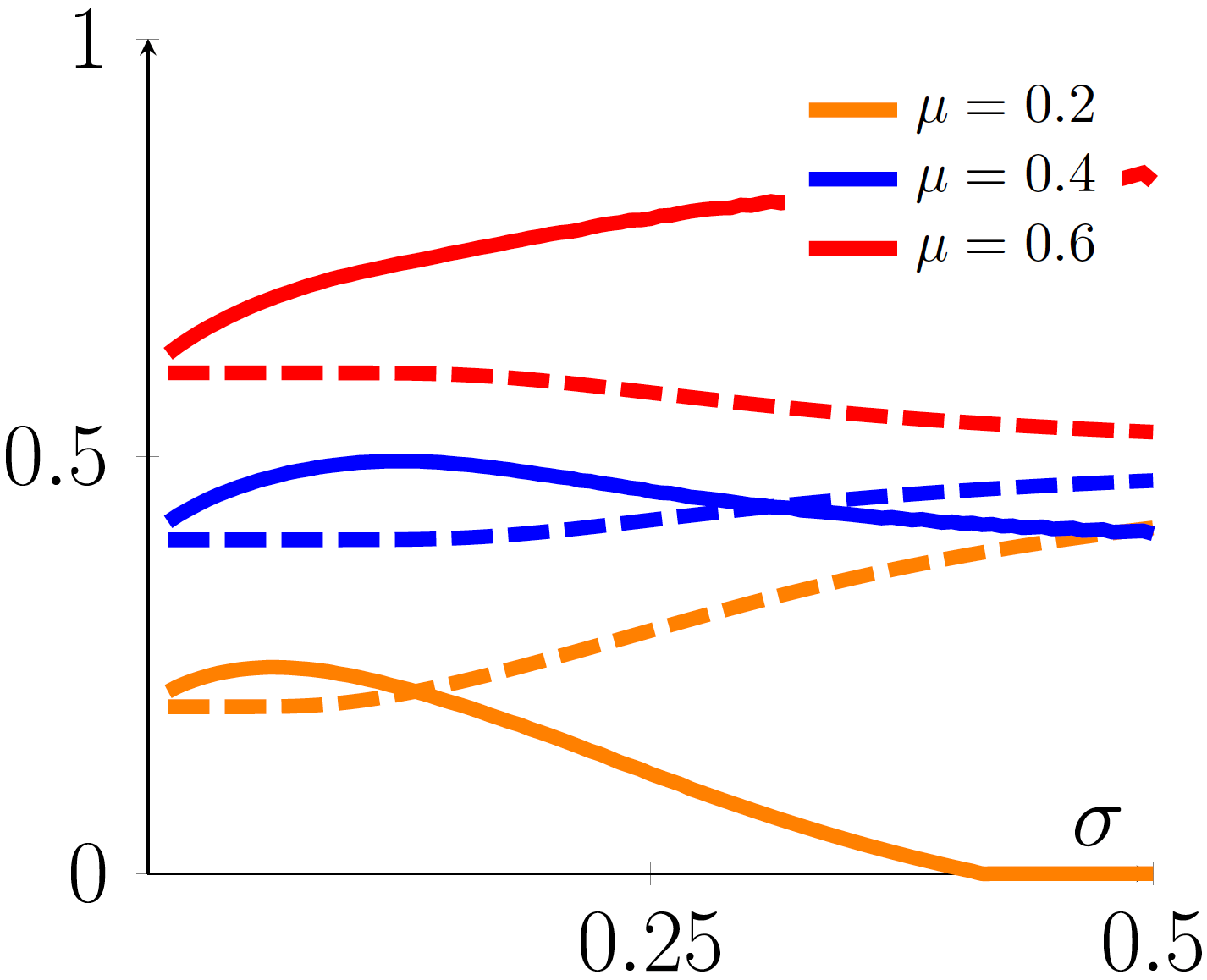}
	\caption{Numerically computed values of $q_i(\beta(0), 0)$ and their corresponding prior means $\pi$ using a $N(\mu, \sigma^2)$ distribution truncated to have support $[0, 1]$. For each color, the solid line is $q_i(\beta(0), 0)$ and the dashed line is the prior mean $\pi$ after truncation, as functions of the generating variance $\sigma$.}
	\label{fig:sims}		
\end{figure}

Full details of the algorithm for computing $q_i(\beta(0), 0)$ are in Appendix \hyperref[sec:c]{C}; Figure \ref{fig:sims} shows an example of the output from these computations. Each color represents a fixed mean $\mu$ of the generating normal distribution, with the $x$-axis representing that distribution's standard deviation $\sigma$. Because the normal distribution is truncated to produce a prior in $[0, 1]$, the ``true mean" $\pi$ of that prior depends on both $\mu$ and $\sigma$; it is shown as a dashed line. The solid line shows the numerically computed value $q_i(\beta(0), 0)$. Thus there is a gap between the small-$r^*$ case of Proposition \ref{prop:contSmall} and the large-$r^*$ case of Proposition \ref{prop:contLarge} if and only if a dashed line lies above its corresponding same-color solid line. When $\mu > 1/2$, this property never holds and there is no gap between Proposition \ref{prop:contSmall} and Proposition \ref{prop:contLarge}. When $\mu < 1/2$, there is no gap for $\sigma$ small enough, but a gap arises for larger $\sigma$. However, making $\sigma$ too large violates the assumption $f'(0) < 1 - 2\pi$, invalidating the propositions. These results suggest that double-truncated uniform distributions are optimal for many possible priors.

\section{Extensions}
\label{sec:ext}
In the motivating example, a politician has a well-defined prior belief about the state of the world and knows the average voter's cost-effectiveness threshold, but she makes no further assumptions on the distribution. This difference in information is not unreasonable: the politician can fine-tune the details of her welfare program, but voter preferences are subject to a number of factors outside her control, e.g., opposition campaigning and news coverage. Loosely speaking, limited data about voter preferences allows the politician to estimate the population mean with convergence rate $1/n$, but estimation of the distribution (or any given quantile) converges at rate $1/\sqrt{n}$; she may thus be more willing to base her strategy on the former than the latter. Despite these justifications, it may still be realistic to weaken these informational assumptions; I do so in this section and discuss how my existing results extend.\footnote{I am grateful to an anonymous referee for suggesting these approaches.}

\subsection{Doubly-Maxmin Sender Preferences}
I first consider a Sender who knows only the mean state and mean Receiver type, with maxmin preferences over all possible pairs of independent distributions fitting those moment restrictions. In the binary-state setting, the mean state and state distribution are equivalent. In the continuous-state case, allowing Nature to choose both distributions in response to Sender's choice of information structure makes Sender's optimization problem ill-posed since her set of feasible distributions depends on the prior distribution. There are two straightforward ways to ensure this set is well-defined. First, I can alter the order of the moves: Nature chooses a unimodal prior distribution (as described in Section \ref{sec:cont}), then Sender chooses a distribution of posterior means, and finally Nature chooses a Receiver type distribution with the known mean. In this case, my characterizations of Sender's optimal distribution (Propositions \ref{prop:contSmall} and \ref{prop:contLarge}) apply to each possible prior. The overall solution follows by minimizing (over possible priors) Sender's maxmin utility from the existing model. Given the lack of closed-form solutions in Section \ref{sec:cont}, this minimization requires a numerical approach. Second, I can ignore the prior altogether by letting Sender choose a distribution of posterior means, then let Nature simultaneously choose any mean-preserving spread of that distribution (which I can call a ``prior") and any Receiver type distribution with the known mean. Since only the distribution of posterior means affects Sender's utility, this approach reduces to the binary-state setting of my model where Sender faces a mean constraint.

\subsection{Flexible Mean Receiver Types} If Nature may freely select from among different mean Receiver types, the highest permissible mean Receiver type is always worst for Sender. However, I can partially relax the mean restriction by allowing Nature to choose higher mean Receiver types only by paying some cost (similar to the variational preference setup in \citealt{MMR2006}). If Nature is restricted to choosing a single mean Receiver type, then the results of my model are unchanged for that fixed mean and its value for Nature can be computed either analytically using Proposition \ref{prop:bin} for the binary-state case or numerically in the continuous-state case. If Nature may randomize over mean Receiver types, then I can solve the binary-state case using the characterization in Proposition \ref{prop:bin}.

Let $u_s(r^*, r)$ be Sender's utility if she chooses the optimal distribution for mean Receiver type $r^*$, but in fact faces a realized mean Receiver type $r$ (drawn from some distribution chosen by Nature). For any $r^*$, the value of $u_s(r^*, r)$ is given by the concavification of a truncated uniform distribution (following Proposition \ref{prop:bin} and Lemma \ref{lm:cav}), so it is linear in $r$.\footnote{It is linear if $r$ is in the support of Sender's optimal distribution. Any mass strictly above the support brings Nature no additional benefit, so I assume without loss that this choice is never made.} Thus if the expected mean Receiver type chosen by Nature is $r^{**}$, Sender's utility from choosing the $r^{*}$-optimal distribution is $u_s(r^*, r^{**})$ regardless of the full distribution of mean Receiver types. This expression is maximized by choosing the $r^{**}$-optimal posterior distribution, in which case Sender could do no better even if she knew the mean Receiver type was $r^{**}$ with certainty. Therefore, fixing some zero-cost mean Receiver type $r^*$, Nature's gain from any distribution with mean Receiver type $r^{**}$ is captured by the difference between Sender's maxmin utility at $r^*$ and her maxmin utility at $r^{**}$. Sender's maxmin utility is a well-defined, continuously differentiable function of the mean Receiver type (computed using Proposition \ref{prop:bin}), so Nature's choice of $r^{**}$ can be straightforwardly found by, e.g., setting the marginal cost of increasing the expected mean Receiver type equal to the marginal decrease in Sender's utility. Sender's optimal distribution when facing this new mean Receiver type is given by using $r^{**}$ in Proposition \ref{prop:bin}.

In the continuous-state case, the presence of multiple optimal distributions (some of which do not have linear concavifications), the kink in the concavification of a double-truncted uniform distribution (as discussed in Section \ref{subsec:nonDTU}), and the potential gap between Propositions \ref{prop:contSmall} and \ref{prop:contLarge} rule out the approach above. Finding Sender's optimal response in this case would require a different characterization of her maxmin utility, so I leave it for future work. 

\section{Conclusion}
\label{sec:conc}
Bayesian persuasion provides a tractable model of communication that can be extended to include rich uncertainty about the Receiver who is the target of persuasion. This work contributes to a growing literature that also introduces ambiguity by posing a maxmin persuasion problem, where the Sender seeks to be robust to any possible prior belief about Receiver types with a known mean. In a binary-state setting, I show a connection to mean-preserving contraction (MPC) games, where competing players choose mean-preserving contractions of probability distributions to obtain the highest realization, and fully characterize Sender's optimal distribution. As in many other MPC games, when her constraint is strong enough Sender chooses a uniform distribution mixed with atoms at the lowest and highest posterior beliefs. These results highlight the importance of the tie-breaking assumption in persuasion problems and emphasize the strength of the maxmin criterion, which delivers strictly lower utility for Sender than any prior belief over Receiver types when the probability of the high state is less than $1/2$. I then use a geometric approach to show, in both a finite-support setting and a novel continuous-state setting, that uniform distributions with an atom at the lower bound of their support are in many cases still optimal. Unlike in the binary-state setting, these distributions now have support in the interior of $[0, 1]$, and Sender's optimal distribution is no longer unique. However, the intuition of linearizing the prior belief over states in order to make Nature indifferent between many worst-case Receiver type distributions is preserved.

\bibliographystyle{ecta}
\bibliography{mmPers.bib}

\newpage
\section*{Appendix A: Omitted Proofs for Section \ref{sec:bin}}
\label{sec:a}
\subsection*{A1: Maxmin Persuasion and MPC Games}
\label{subsec:a1}
I first show that any Nash equilibrium strategy for Sender in the MPC game described in Section \ref{sec:bin} solves the maxmin persuasion problem of Equation (\ref{eq:gen}):
\begin{lemma}
	\label{lm:saddle}
	Consider the MPC game with tie-breaking against Sender, where Sender's choice set is defined as
	\begin{equation*}
		\mathcal{G} = \left\{\text{cdf } G \text{ over } [0, 1]\, | \, G \text{ is an MPC of } F\right\},
	\end{equation*}
	and Nature's choice set is defined as 
	\begin{equation*}
		\mathcal{T} = \left\{\text{cdf } T \text{ over } [0, 1]\, | \, T \text{ is an MPC of } (1 - r^*) \, \delta_0 + r^* \delta_1 \right\}.
	\end{equation*}
	$G^* \in \mathcal{G}$ is a Nash equilibrium strategy for Sender if and only if $G^*$ solves the maxmin persuasion problem of Equation (\ref{eq:gen}).
\end{lemma}
\begin{proof}
	I first show that Sender's utility from the maxmin persuasion problem is the same as from an analogous minmax problem. By Proposition 1 of \citet{KMS2020}, both $\mathcal{G}$ and $\mathcal{T}$ are compact and convex. Because this result uses the norm topology, both spaces are metric spaces, hence Hausdorff spaces. 
	
	The functional of Equation (\ref{eq:gen}) is linear in both distributions, so it is convex in $G$ and concave in $T$. For fixed $G \in \mathcal{G}$, it is also lower semicontinuous on $T$: the result follows from lower semicontinuity of the indicator function (which applies for $q > r$, any $q \in [0, r]$ produces the same value) and application of Fatou's Lemma.
	 
	Therefore, I can apply Theorem 2 of \citet{F1953} to state that
	\begin{equation*}
		\sup_{G \in \mathcal{G}} \left\{\min_{T \in \mathcal{T}} \int \int \textbf{1}(q > r) \, dG(q) \, dT(r) \right\} = \min_{T \in \mathcal{T}} \left\{\sup_{G \in \mathcal{G}} \int \int \textbf{1}(q > r) \, dG(q) \, dT(r) \right\}.
	\end{equation*}
	Because $\mathcal{G}$ is compact, I can in fact replace the supremum on the left-hand side of the equation with a maximum, giving precisely the expression in Equation (\ref{eq:gen}).\footnote{I cannot replace the supremum on the right-hand side with a maximum, and indeed the results for MPC games are often stated using limits of sequences of distributions.} It is then clear that Sender's utility with simulatenous moves in the MPC game must be equal to her utility in the maxmin persuasion problem.
	
	I now prove the ``if" portion of the lemma. Let $G^*$ solve the maxmin persuasion problem. Since the maxmin and minmax utilities for Sender are equal, the game has a value, and both Sender and Receiver have strategies that guarantee them at least the value. $G^*$ is by definition such a strategy; let $T^*$ be such a strategy for Receiver. It must be that the pair $(G^*, T^*)$ guarantees each player exactly the value of the game because it is zero-sum: if either player's utility were strictly above the value, then the other's would be strictly below it. Thus $G^*$ is a best response to $T^*$, since no other strategy gives Sender strictly higher utility (or $T^*$ would not guarantee Receiver the value). Similarly, $T^*$ is a best response to $G^*$. Thus $G^*$ is a Nash equilibrium in the MPC game for Sender.
	
	For the ``only if" portion, let $G^*$ be a Nash equilibrium strategy in the MPC game for Sender. Then the Nash equilibrium payoff for Sender results from taking Nature's best response to $G^*$. Since Nature's payoff is the opposite of Sender's, that payoff is therefore Sender's minimum utility from $G^*$. Thus a Nash equilibrium distribution for Sender has the same payoff in the MPC game and the maxmin persuasion problem, and that utility is precisely equal to Sender's maximum utility in the maxmin persuasion problem, so $G^*$ solves the maxmin persuasion problem by definition.
\end{proof}
This equivalence does not rely on the support of the prior $F$ in the maxmin persuasion problem, and thus suggests that results from other maxmin persuasion models may be applied to solve richer MPC games.

I next show that different tie-breaking assumptions in the persuasion context can be re-interpreted as restrictions on the domain of distributions each player can choose in the MPC game context.
\begin{lemma}
	\label{lm:tb}
	Consider a two-player MPC game $G$ where the support of Sender's and Receiver's chosen distributions must lie weakly above a common lower bound and weakly below a common upper bound, and tie-breaking selects Receiver as the winner if there is a tie. This game is equivalent to a game $F$ which is identical to $G$ except for the following two changes:
	\begin{enumerate}
		\item The support of Receiver's distribution is not bounded above\textemdash Receiver may choose any mean-preserving contraction that obeys the lower bound.
		\item Tie-breaking is even\textemdash in the case of a tie, a winner is randomly chosen.
	\end{enumerate}
\end{lemma}
\begin{proof}
	The game $F$ is (modulo simultaneous moves, which Lemma \ref{lm:saddle} shows are irrelevant) the same as the maxmin persuasion problem. There, Sender persuades Receiver type $r < 1$ by generating a posterior $q_r^0 = r$; with unfavorable tie-breaking, she must generate $q_r^\varepsilon = r + \varepsilon$ for arbitrary $\varepsilon > 0$. Since posteriors must lie in $[0, 1]$, posteriors $q_1^\varepsilon$ are infeasible and Sender can never persuade Receiver type $r = 1$. As $\varepsilon \rightarrow 0$, the effect on the Bayes-plausibility constraint from replacing any $q_r^0$ with $q_r^\varepsilon$ vanishes, allowing a Sender facing unfavorable tie-breaking to match her utility with favorable tie-breaking (and thus for any intermediate tie-breaking rule) for interior Receiver types, but not for type $r = 1$. Thus Sender's utility is affected by the tie-breaking rule if and only if she chooses a posterior distribution with an atom at $q = 1$.
	
	The MPC game $G$, where the chosen distributions must have support in $[0, 1]$ for Sender and in $[0, \infty)$ for Receiver, means that even under favorable tie-breaking for Sender, Nature can generate Receiver types $r_q^\varepsilon = q + \varepsilon$ and keep Sender's utility to the same level as with Receiver-favoring tie-breaking. In particular, Nature can generate type $r_1^\varepsilon = 1 + \varepsilon$, which it would not be able to do if constrained by the upper bound. Thus, a Nash equilibrium of $G$ is equivalent to one of $F$.
\end{proof}

Since tie-breaking against Sender allows me to work with a well-defined minimizing Receiver type distribution for each posterior distribution, I choose this rule. Thus, combining both lemmas, I may apply existence and uniqueness results from MPC games with arbitrary tie-breaking rules and choice set
\begin{equation*}
	\mathcal{G} = \left\{\text{cdf } G \text{ over } [0, 1]\, | \, G \text{ is an MPC of } F\right\},
\end{equation*}
for Sender and 
\begin{equation*}
	\mathcal{T}' = \left\{\text{cdf } T \text{ over } \mathbb{R}_+\, | \, T \text{ is an MPC of } (1 - r^*) \, \delta_0 + r^* \delta_1 \right\}.
\end{equation*}
for Nature.

\subsection*{A2: Characterizing Sender's Optimal Distribution}
\label{subsec:a2}
Given the equivalence result of the prior section, a Nash equilibrium of the MPC game in \citet{H2015}, which the author calls a ``Captain Lotto game," provides a solution to the maxmin persuasion problem of Equation (\ref{eq:gen}). Thus the strategy of Player B in Theorem 4 of that work now gives an optimal posterior distribution for Sender. Theorems 4 and 5 of \citet{NA2018} show that the Nash equilibrium strategy for Player B in the Captain Lotto game is unique when $\pi \leq 1/2$, and therefore so is Sender's optimal posterior distribution. To complete the proof of Proposition \ref{prop:bin}, I replace the sufficient condition for Nash equilibrium when $\pi > 1/2$ in Theorem 10 of \citet{NA2018} with a necessary and sufficient condition for optimality of Sender's chosen posterior distribution.

I begin by showing that Sender's utility can be expressed as a function of $\bar{G}$, the concavification of $G$:
\begin{lemma}
	\label{lm:cav}
	Consider the maxmin persuasion problem of Equation (\ref{eq:gen}) and let $\bar{G}: [0, 1] \rightarrow [0, 1]$ be the concavification of $G$, i.e., the infimum over the set of concave functions $H: [0, 1] \rightarrow [0, 1]$ satisfying
	\begin{equation*}
		H(q) \geq G(q) \hspace{0.5em} \forall \hspace{0.5em} q \in [0, 1].
	\end{equation*}
	Then the following equality holds:
	\begin{equation*}
		\min_{T \in \mathcal{T}} \int \int \textbf{1}(q > r) \, dG(q) \, dT(r) = 1 - \bar{G}(r^*).
	\end{equation*}
\end{lemma}
\begin{proof}
	Manipulating the bounds of integration to rewrite Sender's objective function from Equation (\ref{eq:gen}) gives
	\begin{equation*}
		\begin{split}
			\int_{[0, 1]} \bigg(\int_{[0, 1]} \textbf{1}(q > r) \, dG(q)\bigg) \, dT(r) 
			& = \int_{[0, 1]} \bigg(\int_{[r, 1]} 1 \, dG(q)\bigg) \, dT(r) 
			\\& = \int_{[0, 1]} (1 - G(r)) \, dT(r).
		\end{split}
	\end{equation*}
	Then the minimzation portion of the problem can be written as 
	\begin{equation*}
		\max_{T \in \Delta([0, 1])} \int G(r) \, dT(r) \hspace{0.25cm} \text{ s.t. } \int r \, dT(r) = r^*,
	\end{equation*}
	where I have dropped the constant, rewritten the min as a max, and explicitly included the mean restriction to highlight the similarity to a Bayesian persuasion problem. In this case, the Receiver type $r$ fills the role of ``posterior belief," Nature's utility from a realized Receiver type is $G(r)$, and the ``prior" is the distribution with support $\left\{0, 1\right\}$ and mean $r^*$. This final point follows from the observation in Section \ref{sec:model} that when the prior distribution has binary support, the Bayes-plausibility constraint is the same as a mean restriction. Thus by Corollary 2 of \citet{KG2011}, Nature's utility is given by $\bar{G}$. the concavification of $G$ over the interval $[0, 1]$, evaluated at the prior mean $r^*$. Flipping the sign again, Sender's utility is $1 - \bar{G}(r^*)$.
\end{proof}

My necessary and sufficient condition is an immediate consequence of this result:
\begin{lemma}
	\label{lm:bigP}
	Let $\pi > 1/2$. Then a posterior distribution $G^*$ is optimal for Sender if and only if $\mathbb{E}_{G^*}[\omega] = \pi$ and $G^*(q) \leq q \hspace{0.5em} \forall \hspace{0.5em} q \in [0, 1]$. More than one distribution satisfying this condition always exists.
\end{lemma}
\begin{proof}
	As established in Section \ref{sec:model}, the only constraint on a feasible distribution $G$ for Sender is that $\mathbb{E}_{G}[\omega] = \pi$; I show that the second constraint is both necessary and sufficient for optimality.
	
	Assume $G^*(q) \leq q \hspace{0.5em} \forall \hspace{0.5em} q \in [0, 1]$. Then the function $U(q) = q$ upper-bounds $G^*$ and is concave. Since $U$ is the pointwise-smallest concave function on $[0, 1]$ passing through the point $(1, 1)$, it must therefore be the concavification of $G^*$, and Sender's utility from $G^*$ is $1 - U(r^*) = 1 - r^*$. Because Nature may always choose a Receiver type distribution $T$ with $\supp(T) = \left\{0, 1\right\}$, Sender's utility from any posterior distribution is no more than $1 - r^*$ (the probability that Receiver type $r = 0$ is drawn from $T$). Thus $G^*$ attains the upper bound and is optimal for Sender. There are at least two such distributions for any $\pi > 1/2$. The first is given by solving $\pi = n/(n + 1)$ for $n$ and setting $G^*(q) = q^n$. The second is given by 
	\begin{equation*}
		G^*(q) = \begin{cases}
			0 & q \in [0, 2\pi - 1),\\
			(q + 1 - 2\pi)/(2 - 2\pi) & q \in (2\pi - 1, 1].
		\end{cases}
	\end{equation*}
	Therefore an optimal distribution always exists and is non-unique.
	
	Now assume $G^*$ is optimal for Sender; then, since I have just shown an optimal distribution exists, it must be that $1 - \bar{G}^*(r^*) = 1 - r^*$. But the only weakly positive concave function $H$ on $[0, 1]$ satisfying $H(1) = 1$ and $H(r^*) = r^*$ is $U(q) = q$. Any distinct concave function must have slope greater than 1 at $r^*$\textemdash any less and it would fail to pass through the point $(1, 1)$\textemdash and must therefore have $H(0) < 0$. Therefore $\bar{G}^* = U$ and $G^*(q) \leq q \hspace{0.5em} \forall \hspace{0.5em} q \in [0, 1]$.
\end{proof}
This lemma completes the proof of Proposition \ref{prop:bin}. However, the theorems I reference rely on lengthy computations of Sender's utility under different strategy profiles. In the next section, I provide a clearer geometric proof that does not rely on MPC games and instead highlights the usefulness of concavification.

\subsection*{A3: An Alternative Proof of Proposition \ref{prop:bin}}
\label{subsec:a3}
The concavification result of Lemma \ref{lm:cav} means that Sender's utility from any posterior distribution $G$ is a convex function of $r^*$. It can therefore be lower-bounded by a line tangent to that function through the fixed $r^*$ in the maxmin persuasion problem. The key step of my alternative proof of Proposition \ref{prop:bin} is to show that if $G$ gives Sender a higher utility than the optimal distribution $G^*$, then that tangent line implicitly defines a cdf whose mean is greater than $\pi$. Because the tangent line lies below the function $1 - G$, it must therefore be that $G$ itself has a mean greater than $\pi$, and thus $G$ is not a Bayes-plausible posterior distribution.

Towards establishing this result, consider \textit{upper-truncated uniform posterior distributions} (henceforth UTUs), a class of posterior distributions which place mass $x \geq 0$ on posterior $q = 0$, equal mass on all posteriors $q \in (0, r_h]$ for some $r_h \leq 1$, and no mass on posteriors $q \in (r_h, 1]$. I can use Bayes-plausibility to solve for the unique value of $r_h$ corresponding to a given $x$, so that a UTU is fully characterized by $x$:
\begin{equation*}
	\begin{split}
		\pi & = \int q \, dG_x(q) = \int 1 - G_x(q) \, dq = \frac{(1 - x) \, r_h(x)}{2}
		\\& \Leftrightarrow r_h(x) = \frac{2\pi}{1 - x},
	\end{split}
\end{equation*}
Since $r_h$ is uniquely determined by $x$, I denote a UTU by $G_x$. The following lemma shows that a single choice of $x$ is optimal among all UTUs and can be written as a closed-form function of $r^*$: 
\begin{lemma}
	\label{lm:utu}
	Let $\pi \leq 1/2$. Then if $r^* \leq \pi$, Sender's unique optimal UTU is $G_0$; if $\pi \leq r^* \leq 1/2$, it is $G_{1 - \pi/r^*}$; and if $1/2 \leq r^*$ it is $G_{1 - 2\pi}$.
\end{lemma}
\begin{proof}
	By construction, any UTU $G_x$ is concave and is therefore equal to its concavification $\bar{G}_x$. By Lemma \ref{lm:cav}, the utility from a UTU $G_x$ is therefore
	\begin{equation*}
		1 - \bar{G}_x(r^*) = 1 - G_x(r^*) = \left\{(1 - x) - r^* \frac{(1 - x)^2}{2\pi}\right\}_+.
	\end{equation*}
	The first-order condition in $x$ for the expression in brackets is
	\begin{equation*}
		-1 + r^* \, \frac{1 - x}{\pi} = 0 \Leftrightarrow x_{\text{FOC}} = 1 - \frac{\pi}{r^*}.
	\end{equation*}
	The bracketed expression is increasing in $x$ when $x < x_{\text{FOC}}$ and decreasing in $x$ when $x > x_{\text{FOC}}$. Since $x \in [0, 1 - 2\pi]$, if $r^* < \pi$ the constrained optimal solution is $x^* = 0$ and if $r^* > 1/2$ the constrained optimal solution is $x^* = 1 - 2\pi$; otherwise the optimum is the interior solution $x^* = x_{\text{FOC}} = 1 - \pi/r^*$.
\end{proof}

I now prove two lemmas describing the relationship between the UTU $G_{1 - 2\pi}$ and the function $1 - \bar{G}$ derived from an arbitrary posterior distribution $G$. The first establishes that if, for some posterior distribution $G$, the function $1 - G$ falls below $1 - G_{1 - 2\pi}$ at some mean Receiver type $q$, Sender's utility from $G$ remains below her utility from $G_{1 - 2\pi}$ for all higher Receiver types: 
\begin{lemma}
	\label{lm:bound1}
	Let $G$ be a cdf on $[0, 1]$. Then if there is $q \in [0, 1)$ such that
	\begin{equation*}
		1 - G(q) < 1 - G_{1 - 2\pi}(q),
	\end{equation*}
	then it is also the case that
	\begin{equation*}
		1 - \bar{G}(q') < 1 - G_{1 - 2\pi}(q') \hspace{0.5em} \forall \hspace{0.5em} q' \in [q, 1).
	\end{equation*}
\end{lemma}
\begin{proof}
	The proof is by contradiction. Assume there is $q$ such that
	\begin{equation*}
		1 - \bar{G}(q) \leq 1 - G(q) < 1 - G_{1 - 2\pi}(q),
	\end{equation*}
	but that there is $q' \in [q, 1)$ such that
	\begin{equation*}
		1 - \bar{G}(q') \geq 1 - G_{1 - 2\pi}(q').
	\end{equation*}
	Since $1 - \bar{G}(q) < 1 - G_{1 - 2\pi}(q)$ but $1 - \bar{G}(q') \geq 1 - G_{1 - 2\pi}(q')$, it must be that there is $q_1 \in [q, q']$ where the slope of $1 - \bar{G}$ is strictly greater than that of $1 - G_{1 - 2\pi}$. But because $G$ and $G_{1 - 2\pi}$ are cdfs and $1 - \bar{G}$ is weakly positive, 
	\begin{equation*}
		1 - \bar{G}(1) = 0 = 1 - G(1) = 1 - G_{1 - 2\pi}(1),
	\end{equation*}
	so there must be $q_2 \in [q', 1]$ where the slope of $1 - \bar{G}$ is weakly less than that of $1 - G_{1 - 2\pi}$. Then $q_1 \leq q_2$ but the slope of $1 - \bar{G}$ at $q_1$ is strictly greater than at $q_2$, violating convexity of $1 - \bar{G}$, and thus concavity of $\bar{G}$.
\end{proof}

The next lemma describes features of $1 - \bar{G}$ when the posterior distribution $G$ weakly improves on Sender's utility from $G_{1 - 2\pi}$:
\begin{lemma}
	\label{lm:bound2}
	If $G \neq G_{1 - 2\pi}$ is a cdf such that
	\begin{equation*}
		1 - \bar{G}(r^*) \geq 1 - G_{1 - 2\pi}(r^*) \text{ and } \int q \, dG(q) = \pi,
	\end{equation*}
	then the slope\footnote{Because $1 - \bar{G}$ is convex, it is continuous on $(0, 1)$ and its left and right derivatives are always well-defined. The function $1 - G$ for any UTU $G$ is also continuous with well-defined left and right derivatives. When referring to the slope or to a tangent line I consider the right derivative.} of $1 - \bar{G}$ at $r^*$ is strictly less than the slope of $1 - G_{1 - 2\pi}$ at $r^*$.
\end{lemma} 
\begin{proof}
	I first show that there is $q_d \in (r^*, 1]$ such that 
	\begin{equation*}
		1 - \bar{G}(q_d) \leq 1 - G(q_d) < 1 - G_{1 - 2\pi}(q_d).
	\end{equation*}
	Note that for $G$ to be distinct from $G_{1 - 2\pi}$, there must be some posterior $q_d \in [0, 1]$ where $1 - G(q_d) \neq 1 - G_{1 - 2\pi}(q_d)$. It cannot be the case that 
	\begin{equation*}
		1 - G(q) \geq 1 - G_{1 - 2\pi}(q) \hspace{0.5em} \forall \hspace{0.5em} q \in [0, 1]\ \text{ and } 1 - G(q_d) > 1 - G_{1 - 2\pi}(q_d).
	\end{equation*}
	If that is the case, then because $G$ is a cdf, it is right-continuous, and therefore fixing $\varepsilon > 0$ there is $\delta(\varepsilon) > 0$ such that 
	\begin{equation*}
		1 - G(q') > 1 - G(q_d) - \varepsilon \hspace{0.5em} \forall \hspace{0.5em} q' \in [q_d, q_d + \delta(\varepsilon)).
	\end{equation*}
	Since the slope of $1 - G_{1 - 2\pi}$ is no greater than 0, setting $\varepsilon \in (0, G_{1 - 2\pi}(q_d) - G(q_d))$ ensures that 
	\begin{equation*}
		1 - G(q') > 1 - G_{1 - 2\pi}(q_d) \geq 1 - G_{1 - 2\pi}(q') \hspace{0.5em} \forall \hspace{0.5em} q' \in [q_d, q_d + \delta(\varepsilon)).
	\end{equation*}
	Therefore there is a non-degenerate interval where $1 - G > 1 - G_{1 - 2\pi}$, and by assumption $1 - G \geq 1 - G_{1 - 2\pi}$ everywhere on $[0, 1]$, so integrating the inequality gives a violation of Bayes-plausibility:
	\begin{equation*}
		\int q \, dG(q) = \int 1 - G(q) \, dq > \int 1 - G_{1 - 2\pi}(q) \, dq = \int q \, dG_{1 - 2\pi}(q) = \pi.
	\end{equation*}
	Thus by contradiction there must be $q_d \in [0, 1]$ such that
	\begin{equation*}
		1 - \bar{G}(q_d) \leq 1 - G(q_d) < 1 - G_{1 - 2\pi}(q_d).
	\end{equation*}
	By Lemma \ref{lm:bound1}, since $1 - \bar{G}(r^*) \geq 1 - G_{1 - 2\pi}(r^*)$, there is no $q \in [0, r^*)$ where $1 - G(q) < 1 - G_{1 - 2\pi}(q)$. Thus it must be that 
	\begin{equation*}
		1 - G(q) \geq 1 - G_{1 - 2\pi}(q) \hspace{0.5em} \forall \hspace{0.5em} q \in [0, r^*],
	\end{equation*}
	and therefore $q_d \in (r^*, 1]$.
	
	The claim now follows by the argument in Lemma \ref{lm:bound1}. Since $1 - \bar{G}(r^*) \geq 1 - G_{1 - 2\pi}(r^*)$ and $1 - \bar{G}(q_d) < 1 - G_{1 - 2\pi}(q_d)$, there is $q' \in [r^*, q_d]$ where the slope of $1 - \bar{G}$ is strictly less than that of $1 - G_{1 - 2\pi}$. But since $\bar{G}$ is concave, $1 - \bar{G}$ is convex and its slope cannot increase as $q$ decreases; the slope of $1 - \bar{H}$ at $r^*$ must therefore be strictly less than that of $1 - G_{1 - 2\pi}$ at $r^*$. 
\end{proof}
The implication is vacuous for $r^* \leq 1/2$, where there are no posterior distributions that meet the conditions; however, even in that case the result is central to a proof by contradiction. 

With these three lemmas in hand, I now provide an alternative proof of the case $\pi \leq 1/2$ in Proposition \ref{prop:bin}:
\begin{lemma}
	\label{lm:bin}
	If $\pi \leq 1/2$, Sender's unique optimal posterior distribution is as follows:
	\begin{itemize}
		\item If $r^* \leq \pi \leq 1/2$,
		\begin{equation*}
			G^*(q) = U[0, 2\pi].
		\end{equation*}
		\item If $\pi \leq r^* \leq 1/2$,
		\begin{equation*}
			G^*(q) = \bigg(1 - \frac{\pi}{r^*}\bigg) \, \delta_0 + \frac{\pi}{r^*} \, U[0, 2r^*].
		\end{equation*}
		\item If $\pi \leq 1/2 \leq r^*$,
		\begin{equation*}
			G^*(q) = (1 - 2\pi) \, \delta_0 + 2\pi \, U[0, 1].
		\end{equation*}
	\end{itemize}
\end{lemma}
\begin{proof}
	The proof is by contradiction. Let $G$ be a proposed alternative posterior distribution that delivers weakly greater utility for Sender than $G^*$. By Lemma \ref{lm:cav} (to define the utility from each posterior distribution) and Lemma \ref{lm:utu} (since $G^*$ is a UTU, it must be uniquely optimal among UTUs), it is the case that
	\begin{equation*}
		1 - \bar{G}(r^*) \geq \bar{G}^*(r^*) = 1 - G^*(r^*) \geq 1 - G_{1 - 2\pi}(r^*).
	\end{equation*}
	Consider the line $L$ that is tangent to $1 - \bar{G}$ at $r^*$.\footnote{Recall that if $r^*$ is a kink point of $1 - \bar{G}$, I use the right derivative of $1 - \bar{G}$ to define the slope.} Because $\bar{G}$ is convex and weakly positive (recall that the line $\ell(q) = 0$ is convex and lower-bounds $1 - G$), it is lower-bounded by $L_+(q) = \max\left\{L(q), 0\right\}$. Furthermore, by Lemma \ref{lm:bound2}, the slope of $L$ is less than that of $1 - G_{1 - 2\pi}$, so it must be that 
	\begin{equation*}
		1 \geq 1 - G(0) \geq 1 - \bar{G}(0) \geq L_+(0) > 1 - G_{1 - 2\pi}(0) = 2\pi.
	\end{equation*}
	
	For any $x \in [0, 1 - 2\pi]$, there is a corresponding UTU $G_x$ with $G_x(0) = x$. Since $1 - L_+(0) \in [0, 1 - 2\pi]$, there exists an UTU\textemdash call it $G_\text{alt}$ for alternative\textemdash with $1 - G_\text{alt}(0) = L_+(0)$. If $G_\text{alt} \neq G^*$, then because $G^*$ is uniquely optimal among UTUs, it must be that
	\begin{equation*}
		L_+(r^*) = 1 - \bar{H}(r^*) \geq 1 - \bar{G^*}(r^*) > \bar{G}_\text{alt}(r^*) = 1 - G_\text{alt}(r^*).
	\end{equation*}
	Then, because $L$ and $1 - G_\text{alt}$ intersect at $q = 0$ but $L$ is greater than $1 - G_\text{alt}$ at $q = r^*$, it must be that the slope of $L$ is strictly greater than the slope of the strictly downward-sloping portion of $1 - G_\text{alt}$; therefore in fact
	\begin{equation*}
		L_+(q) \geq 1 - G_\text{alt}(q) \hspace{0.5em} \forall \hspace{0.5em} q \in [0, 1] \hspace{0.25cm} \text{ and } \hspace{0.25cm} L_+(q') > 1 - G_\text{alt}(q') \hspace{0.5em} \forall \hspace{0.5em} q' \in (0, r^*].
	\end{equation*}
	Integrating the expression and using the fact that $L_+$ lower-bounds $1 - \bar{G}$, which in turn lower-bounds $1 - G$, it is the case that
	\begin{equation*}
		\begin{split}
			\int q \, dG(q) = \int 1 - G(q) \, dq \geq \int 1 - \bar{G}(q) dq & \geq \int L_+(q) \, dq
			\\&  > \int 1 - G_\text{alt}(q) \, dq = \int q \, dG_\text{alt}(q) = \pi.
		\end{split}
	\end{equation*}
	The first and penultimate equalities are both from integration by parts, and the final equality is because all UTUs (including $G_\text{alt}$) are Bayes-plausible by construction. Therefore $G$ violates Bayes-plausibility and is not a valid alternative distribution.  
	
	Even when $G_\text{alt} = G^*$, it is still the case that, whenever 
	\begin{equation*}
		L_+(r^*) = 1 - \bar{G}(r^*) > 1 - \bar{G}^*(r^*) = 1 - G^*(r^*),
	\end{equation*} 
	the slope of $L$ is greater than the slope of the strictly downward-sloping portion of $1 - G^*$. In this case, $L_+(q) > 1 - G^*(q) \hspace{0.5em} \forall q \hspace{0.5em} \in (0, r^*]$ and
	\begin{equation*}
		\begin{split}
			\int q \, dG(q) = \int 1 - G(q) \, dq \geq \int 1 - \bar{G}(q) dq & \geq \int L_+(q) \, dq 
			\\& > \int 1 - G^*(q) \, dq = \int q \, dG^*(q) = \pi,
		\end{split}
	\end{equation*}
	just as before. Thus $G$ again violates Bayes-plausibility. 
	
	If instead $G_\text{alt} = G^*$ but now $L_+(r^*) = 1 - G(r^*)$, it must be the case that $L$ and the strictly downward-sloping portion of $1 - G$ have the same slope, so in fact 
	\begin{equation*}
		L_+(q) = 1 - G^*(q) \hspace{0.5em} \forall \hspace{0.5em} q \in [0, 1].
	\end{equation*}
	Then then there are two possible cases. The first is trivial: 
	\begin{equation*}
		1 - G(q) = L_+(q) = 1 - G^*(q) \hspace{0.5em} \forall \hspace{0.5em} q \in [0, 1],
	\end{equation*}
	so that $G$ is not a deviation at all. In the second, there must be some $q \in [0, 1]$ so that $1 - G(q) > L_+(q)$; recall that $L_+$ lower-bounds $1 - G$, and thus the direction of the inequality is known. Because $G$ is a cdf, it is right-continuous, and therefore fixing $\varepsilon > 0$ there is $\delta(\varepsilon) > 0$ such that 
	\begin{equation*}
		1 - G(q') > 1 - G(q) - \varepsilon \hspace{0.5em} \forall \hspace{0.5em} q' \in [q, q + \delta(\varepsilon)).
	\end{equation*}
	Since the slope of $L_+$ is no greater than 0, setting $\varepsilon \in (0, 1 - G(q) - L_+(q))$ ensures that 
	\begin{equation*}
		1 - G(q') > L_+(q) \geq L_+(q') \hspace{0.5em} \forall \hspace{0.5em} q' \in [q, q + \delta(\varepsilon)).
	\end{equation*}
	Therefore there is a non-degenerate interval where $1 - G > L_+$, and $1 - G \geq L_+$ everywhere on $[0, 1]$, so integrating the inequality gives
	\begin{equation*}
		\int q \, dG(q) = \int 1 - G(q) \, dq > \int L_+(q) \, dq = \int 1 - G^* dq = \int q \, dG^*(q) = \pi,
	\end{equation*}
	as desired. Having covered both the case $G_\text{alt} \neq G^*$ and the case $G_\text{alt} = G^*$, I have shown that in all cases $H$ violates Bayes-plausibility and therefore, by contradiction, $G^*$ is uniquely optimal.
\end{proof}

The full proof of Proposition \ref{prop:bin} without reference to MPC games is therefore obtained by combining Lemmas \ref{lm:bigP} and \ref{lm:bin}.

\subsection*{A4: Results with Alternative Tie-Breaking}
\label{subsec:a4}
The difference in Sender's optimal posterior distribution from the standard Bayesian persuasion problem is clearly driven by Nature's ability to tailor a worst-case Receiver type distribution to Sender's particular disclosure strategy, but may also be affected by the ability to generate a Receiver type who is unconvinced even when the state is surely $\omega = 1$. By following Lemma \ref{lm:tb}, I can obtain Sender's optimal posterior distribution under favorable tie-breaking, and show exactly when the choice of tie-breaking rule is influential:
\begin{corollary}[Adapted from Theorem 4 of \citealt{H2015}]
	\label{cor:tbGood}
	Let $\supp(F) = \left\{0, 1\right\}$ and let ties be broken in favor of Sender.\\
	One optimal posterior distribution for Sender's is as follows:
	\begin{itemize}
		\item If $r^* \leq \pi \leq 1/2$,
		\begin{equation*}
			G^*(q) = U[0, 2\pi].
		\end{equation*}
		\item If $\pi \leq r^* \leq 1/2$,
		\begin{equation*}
			G^*(q) = \bigg(1 - \frac{\pi}{r^*}\bigg) \, \delta_0 + \frac{\pi}{r^*} \, U[0, 2r^*].
		\end{equation*}
		\item If $1/2 < r^*$,
		\begin{equation*}
			G^*(q) = (1 - \pi) \, \delta_0 +  \pi \, \delta_1.
		\end{equation*}
		\item If $r^* \leq 1/2 < \pi$,
		\begin{equation*}
			G^*(q) = (2 - 2\pi) \, U[0, 1] + (2\pi - 1) \, \delta_1.
		\end{equation*}
	\end{itemize}
\end{corollary}
\begin{proof}
	I use the MPC game representation of the maxmin persuasion problem. The result follows directly from Player A's equilibrium strategy in Theorem 4 of \citet{H2015}, where I replace the use of $\varepsilon$-approximating distributions, which are not needed in my setting, with the exact upper bound of 1 on posterior beliefs. Since the alternate tie-breaking rule is not the focus of this work, I do not provide a full characterization of other optimal posterior distributions; as \citet{NA2018} shows, attempting a full characterization through the connection to MPC games becomes complex.
\end{proof}
In this case, Sender sometimes takes advantage of favorable tie-breaking and places an atom at posterior $q = 1$, exploiting Nature's inability to generate a skeptical Receiver type for that posterior belief. This choice allows Sender to obtain a utility higher than $1 - r^*$, since even if the Receiver type is $r = 1$, they are now persuaded whenever posterior $q = 1$ is realized. However, creating this atom tightens the Bayes-plausibility constraint, so if neither Sender nor Nature's constraint is slack enough to allow frequent realizations of 1, Sender uses the same approach as with unfavorable tie-breaking. Thus when the probability of the high state and the mean Receiver type are both small, Sender's maxmin utility remains strictly below her utility with even the most unfavorable prior belief about Receiver types, regardless of whether tie-breaking is favorable or not.

In the maxmin persuasion context, it seems natural to break ties either entirely in favor of or entirely against Sender. Those rules allow me to interpret a Receiver of type $r$ either as the highest Receiver type who is convinced by posterior belief $q = r$ or the lowest Receiver type who is not convinced by that belief, respectively. However, if the MPC game is interpreted as competitive persuasion, as in \citet{BC2015}, then it also seems reasonable to consider breaking ties evenly, so as to favor neither player.\footnote{When the players are persuading a Receiver about a common state of the world, as in \citet{AK2020}, it is reasonable to also require $\pi = r^*$. In \citet{BC2015}, the players are schools convincing a Receiver about the binary ability of a student drawn from a school-specific distribution, so $\pi \neq r^*$ represents one school producing more high-type students on average.} This choice is equivalent to requiring that both distributions have support in $[0, 1]$; in that case the optimal posterior distribution (derived without uniqueness in \citealt{H2015} and with uniqueness in \citealt{BC2015} and \citealt{NA2018}) is as follows:
\begin{corollary}[Lemma 3 of \citealt{BC2015}]
	\label{cor:tbEven}
	Let $\supp(F) = \left\{0, 1\right\}$ and let ties be broken evenly.\\
	The unique optimal posterior distribution $G^*$ for Sender's is as follows:
	\begin{itemize}
		\item If $r^* \leq \pi \leq 1/2$,
		\begin{equation*}
			G^*(q) = U[0, 2\pi].
		\end{equation*}
		\item If $\pi \leq r^* \leq 1/2$,
		\begin{equation*} 
			G^*(q) = \bigg(1 - \frac{\pi}{r^*}\bigg) \, \delta_0 + \frac{\pi}{r^*} \, U[0, 2r^*].
		\end{equation*}
		\item If $1/2 \leq \pi$ and $r^* \leq \pi$,
		\begin{equation*}
			G^*(q) = \frac{1 - \pi}{\pi} \, U[0, 2 - 2\pi] + \frac{2\pi - 1}{\pi} \, \delta_1.
		\end{equation*}
		\item If $1/2 \leq r^*$ and $\pi \leq r^*$,
		\begin{equation*}
			G^*(q) = \bigg(1 - \frac{\pi}{r^*}\bigg) \, \delta_0 + \frac{\pi}{r^*} \, \bigg(\frac{1 - r^*}{r^*} \, U[0, 2 - 2r^*] + \frac{2r^* - 1}{r^*} \, \delta_1\bigg).
		\end{equation*}
	\end{itemize}
\end{corollary}
\begin{proof}
	This result appears verbatim in \citet{BC2015}, with Sender as Player A when $r^* \leq \pi$ and Player B when $r^* \geq \pi$.
\end{proof}

Finally,  note for general interest that in the MPC game when both players' feasible distributions have domain $\mathbb{R}_+$ (and are mean-preserving contractions of binary support distributions), the unique solution is the same as the cases $r^* \leq \pi \leq 1/2$ and $\pi \leq r^* \leq 1/2$ of Corollary \ref{cor:tbEven}, with the relationship between $\pi$ and $r^*$ determining which case applies. The solution when $\pi = r^*$, so that the constraints are symmetric, first appears in \citet{BC1980}, and also appears in \citet{M1993}. The solution for the asymmetric case first appears in \citet{SP2006}, and also appears in \citet{H2008}.

\newpage
\section*{Appendix B: Omitted Proofs for Section \ref{sec:cont}}
\label{sec:b}
\subsection*{B1: Properties of DTUs}
\label{subsec:b1}
To begin, I describe DTUs in more detail. The uniform portion of the DTU (between the lower and upper truncations) has slope $\beta$, which I refer to as the slope of the DTU. The line $L(q) = \beta q + y$, which forms that uniform portion, intersects the vertical axis at $y$; I refer to this value as the intercept of the DTU. To derive a relationship between $\beta$, $y$, and $\ell$, I use the fact that Bayes-plausibility requires $\mathbb{E}_G[\omega] = \pi$. This condition immediately imposes the restriction that $\ell \in [0, \pi]$; using simple geometry to compute the integral of a DTU's cdf and set it equal to $1 - \pi$ shows that
\begin{equation*}
	\beta(\ell, y) = \frac{(\pi - y \ell) - \sqrt{(\pi - y \ell)^2 - \ell^2 \, (1 - y)^2}}{\ell^2}.
\end{equation*}
This expression is continuously differentiable for $\ell \in (0, \pi]$ and $y \in [0, 1)$. Fixing $\ell$, $\beta(\ell, y)$ is injective and decreasing in $y$. Fixing $y$, $\beta(\ell, y)$ is injective and increasing in $\ell$, attaining a maximum of $\beta(\pi, y) = (1 - y)/\pi$. While $\beta(0, y)$ is not defined using the expression above, the limit from the right exists:
\begin{equation*}
	\lim_{\ell \rightarrow 0^+} \beta(\ell, y) = \lim_{\ell \rightarrow 0^+} \frac{(1 - y)^2}{(\pi - y \ell) + \sqrt{(\pi - y \ell)^2 - \ell^2 \, (1 - y)^2}} = \frac{(1 - y)^2}{2\pi}.
\end{equation*}
I thus define $\beta(0, y) = (1 - y)^2/(2\pi)$ explicitly. For $y \in [0, 1 - 2\pi]$, $\beta(0, y)$ is the slope of the UTU with intercept $y$. When $y > 1 - 2\pi$, there is no corresponding UTU; instead, the lower bound of interest is $\beta(\ell, y) = 1 - y$, the slope that satisfies $G(1) = 1$.\footnote{This is the desired lower bound because any cdf $H$ over $[0, 1]$ must satisfy $H = 1$, and I wish to use DTUs to upper-bound other feasible probability distributions.} The assumption $y > 1 - 2\pi$ implies $1 - y \in ((1 - y)^2/(2\pi), (1 - y)/\pi)$, so the lower bound is attained at an interior $\ell \in (0, \pi)$; I call this value $\ell^\text{min}_y$. Because the function $\beta(\ell, y) - (1 - y)$ is continuously differentiable, the Implicit Function Theorem ensures that I can write $\ell^\text{min}_y$ as a continuously differentiable function of $y$.

The concavification of a DTU is easy to compute: so long as the slope of the line through $(0, 0)$ and $(\ell, \beta(\ell) \ell + y)$ is weakly less than $\beta(\ell, y)$, the concavification will be
\begin{equation*}
	\bar{G}_y^\ell(q) = \begin{cases}
		(\beta(\ell, y) \, \ell + y)/\ell & q \in [0, \ell),\\
		G_y^\ell(q) & q \in [\ell, 1].
	\end{cases}
\end{equation*}
That condition is simply
\begin{equation*}
	\frac{\beta(\ell, y) \, \ell + y}{\ell} = \beta(\ell, y) + \frac{y}{\ell} \geq \beta(\ell, y),
\end{equation*} 
which always holds since $y \geq 0$ and $\ell \geq 0$. Thus the concavification of a DTU is composed of two upward-sloping line segments with a kink at $\ell$ and a constant line segment in the region of the upper truncation.

\subsection*{B2: $y$-Optimal DTUs}
\label{subsec:b2}
Given a value of the mean Receiver type $r^*$ and a fixed intercept $y$, I show the existence of a well-defined and unique DTU that provides Sender's highest utility among all DTUs with an intercept of $y$. Since $y$ is fixed, for this section I drop the dependence on $y$ from all functions.
\begin{lemma}
	\label{lm:intOpt}
	Given $r^* \in (0, 1)$ and $y \in [0, 1)$, there is a well-defined DTU $G^{\beta(y)}_y$ with lower truncation length $\ell^*_y$ that maximizes Sender's utility among all Bayes-plausible DTUs with intercept $y$.
\end{lemma}
\begin{proof}
	Let $V_y \subset [0, \pi]$ be the set of $\ell$ such that a DTU with lower truncation $\ell$ and intercept $y$ is Bayes-plausible. I first show that $V_y$ is closed; since it is clearly also bounded, $V_y$ is therefore compact. To do so, I define the function
	\begin{equation*}
		v(x, \ell) = \int_0^x F(q) \, dq - \int_0^x G^\ell_y(q) \, dq
	\end{equation*}
	for some DTU $G^\ell_y$ with intercept $y$ and lower truncation $\ell$. This function captures the value of the Bayes-plausibility integral constraint for $G^\ell_y$ at $x \in [0, 1]$. Clearly $v(0, \ell) = 0$, and $v(1, \ell) = 0$ because $\mathbb{E}_F[\omega] = \mathbb{E}_{G^\ell_y}[\omega] = \pi$.
	
	At any $x$, the integral of $G^\ell_y$ on $[0, x]$ is continuous in $\ell$. This result is obvious for $x \neq \ell$ (since $G_\ell(q)$ is continuous in $\ell$ at those points) and holds for $x = \ell$ because the left and right limits as $x \rightarrow \ell$ are both 0. Therefore $v(x, \ell)$ is also continuous in $\ell$ for fixed $x$, since it depends on $\ell$ only through that integral. If $G^\ell_y$ is not Bayes-plausible, then (since it satisfies $\mathbb{E}_{G^\ell_y}[\omega] = \pi$ by construction) there must be some $x_\text{neg} \in (0, 1)$ for which $v(x_\text{neg}, \ell) < 0$. Because $v(x_\text{neg}, \ell)$ is continuous in $\ell$, there is $\varepsilon > 0$ such that for any $\ell'$ in a $\varepsilon$-neighborhood of $\ell$, $v(x_\text{neg}, \ell') < 0$. Therefore any $G^{\ell'}_y$ is not Bayes-plausible, so $U \subset [0, \pi]$, the set of $\ell$ where Bayes-plausibility fails, is open. Since $V_y = [0, \pi] \setminus U$, it must be that $V$ is closed.
	
	By Lemma \ref{lm:cav}, Sender's utility from a DTU is given by
	\begin{equation*}
		u_S(r^*, \ell) = 1 -
		\begin{cases}
			(\beta(\ell) + y/\ell) \, r^* + y & \ell < r^*,\\
			G^\ell_y(r^*) & \ell \geq r^*.\\
		\end{cases}
	\end{equation*}
	This function is continuous in $\ell$ on $[0, \pi]$. Since $\beta(\ell)$ is continuous in $\ell$ on $[0, \pi]$, each of the two piecewise portions of $u_S$ are clearly continuous in $\ell$; it remains only to check the case $\ell = r^*$. But because the left and right limits as $\ell \rightarrow r^*$ exist (by continuity of each piecewise portion) and are equal (by construction of $u_S$), $u_S$ is continuous at $\ell = r^*$ as well. Therefore the image of $V$ under $u_S$ must be compact, and thus contains a well-defined maximum, which is attained by some (possibly multiple) $\ell \in V$.
\end{proof}

Unlike in the binary-state setting, it is not possible to solve analytically for $G^{\beta(y)}_y$. However, appropriate sufficient conditions can ensure that $G^{\beta(y)}_y$ is both unique and slope-minimizing among Bayes-plausible DTUs with intercept $y$:
\begin{lemma}
	\label{lm:minSlope}
	Fix $y \in [0, 1)$ and $r^* \in (0, 1)$. There is a unique and well-defined DTU $G^{sm}_y$ that has minimal slope among all Bayes-plausible DTUs with intercept $y$. If $y = 0$ or $r^* \in [\pi, 1)$, then the $y$-optimal DTU $G^{\beta(y)}_y$ equals $G^{sm}_y$
\end{lemma}
\begin{proof}
	Fix $y \in (0, 1)$. By Lemma \ref{lm:intOpt}, the set $V_y$ of values of $\ell$ such that $G^\ell_y$ is Bayes-plausible is closed, and the function $\beta(\ell, y)$ is continuous and monotonic in $\ell$ for fixed $y$, so there is a unique $\ell_{sm} \in V_y$ such that $\beta(\ell_{sm}, y) = \inf_{\ell \in V_y} \beta(\ell, y)$. 
	
	Now I show that either of the conditions provided in the lemma are sufficient for the slope-minimizing DTU to be optimal. First fix $y = 0$. Then $\beta(\ell) + y/\ell = \beta(\ell)$, so $u_S(r^*, \ell) = 1 - G(r^*)$; that is, there is no kink at $\ell$ in Sender's utility from DTUs with intercept 0. Thus Sender's utility from $G^\ell_y$ is strictly greater than her utility from $G^{\ell'}_y$ if and only if $\beta(\ell) < \beta(\ell')$. By Lemma \ref{lm:intOpt}, there exists a DTU $G^{\beta(0)}_0$ with lower truncation length $\ell^*_0$ that maximizes Sender's utility among all Bayes-plausible DTUs with intercept 0. No other Bayes-plausible DTU can have a strictly smaller slope, since then it would deliver a strictly higher utility. But no other Bayes-plausible DTU can have the same slope, $\beta(\ell^*_0)$, since there can be no $\ell' \neq \ell^*_0$ where $\beta(\ell) = \beta(\ell^*_0)$. Therefore all other Bayes-plausible DTUs have strictly larger slope, and so $G^{\beta(0)}_0$ satisfies both (1) and (2).
	
	If instead $r^* \in [\pi, 1)$, then similarly $u_S(r^*, \ell) = 1 - G(r^*)$; since $\ell \in [0, \pi]$, $r^*$ surely lies weakly above $\ell$. The argument is then the same; a DTU is utility-maximizing if and only if it is slope-minimizing, Lemma \ref{lm:intOpt} guarantees the existence of a utility-maximizing DTU, and the injectivity of the map from $\ell$ to $\beta(\ell)$ guarantees uniqueness.
\end{proof}

\subsection*{B3: Simplifying the Integral Constraint}
\label{subsec:b3}
Let $U_{r^*}$ be the set of utilities attained by any $y$-optimal DTU:
\begin{equation*}
	U_{r^*} = \left\{u_S(r^*, \ell, y) \, | \, G_\ell^y = G^{\beta(y)}_y \text{ for some } y \in [0, 1) \right\},
\end{equation*}
where I restore the dependence on $y$ in $u_S$, since $y$ is no longer fixed. That set is a subset of $[0, 1]$, and is therefore bounded, so $\sup U_{r^*}$, Sender's supremum utility over all $y$-optimal DTUs (and thus over all DTUs) is well-defined and contained in the closure of $U_{r^*}$. Further restrictions on $F$ and $r^*$ provide sufficient conditions for $U_{r^*}$ to be closed, and thus for the maximum to exist. In order to state these sufficient conditions, I first prove Lemma \ref{lm:contEll}. In this proof, I again drop the dependence on $y$ from all functions since $y$ is fixed, but note important changes in the argument for different values of $y$.
\begin{lemma}
	\label{lm:contEll}
	Let $G^{sm}_y$ be the DTU with the minimal slope among all Bayes-plausibile DTUs with intercept $y$, and let $\ell^{sm}_y$ be its lower truncation length. If $y \in [0, 1 - 2\pi]$, then the minimal interior $q$ where $G^\beta_y(q) = F(q)$, call it $q_1(\ell, y)$, is well-defined and $\ell^{sm}_y$ satisfies
	\begin{equation*}
		\begin{split}
			&\int_0^{q_1(\ell^{sm}_y, y)} F(q) \, dq = \int_0^{q_1(\ell^{sm}_y, y)} G^{sm}_y(q) \, dq 
			\\& \hspace{3cm} \text{and} 
			\\& \int_0^{x} F(q) \, dq > \int_0^{x} G^{sm}_y(q) \, dq \qquad \text{if and only if} \qquad x \in (0, q_1(\ell^{sm}_y, y)) \cup (q_1(\ell^{sm}_y, y), 1).
		\end{split}
	\end{equation*}
	
	If instead $y \in (1 - 2\pi, 1)$, then either the two conditions above hold or $\ell^{sm}_y$ equals the minimum lower truncation length $\ell^\text{min}_y$.
\end{lemma}
\begin{proof}
	By Lemma \ref{lm:minSlope}, there exists a unique minimal-slope Bayes-plausible DTU with intercept $y$.
	
	Because of the shape of $F$, the equation $L(q) = \beta(\ell) \, q + y =  F(q)$ has at most two solutions with $q \in (0, 1]$. In particular, if the slope of $L$ is such that it lies completely above $F$ in $(0, 1]$, then there are no solutions in that interval; if the slope of $L$ is such that it is tangent to $F$, then there is one;\footnote{There is at most one value of $\ell$ such that $\beta(\ell) \, q + y$ is tangent to $F$ in $(0, 1]$.} and if the slope of $L$ is less than that of the tangent to $F$ through $y$, there are two.
	
	Consider a DTU $G^{\beta(\ell)}_y$ with lower truncation length $\ell$. If $L(q) \geq F(q) \hspace{0.5em} \forall \hspace{0.5em} q \in (0, 1]$\textemdash that is, $L$ is either tangent to $F$ at a point $q_t$ or lies entirely above $F$\textemdash then this DTU satisfies Bayes-plausibility. The function $v(x, \ell)$, which gives the value of the Bayes-plausibility integral constraint for $G^{\beta(\ell)}_y$ at some $x \in [0, 1]$, is weakly decreasing whenever $L(x) \geq F(x)$.\footnote{When $L(x) > 1$, $G^\ell_y(x) = 1$ rather than following $L(x)$, but since the line $y = 1$ is an upper bound on $F$ as well, the upper truncation does not affect the behavior of $v(x, \ell)$.} Thus $v(x, \ell)$ is weakly decreasing for all $x \in (\ell, 1)$. Since $v(1) = 0$, it must therefore be that $v(x, \ell) \geq 0 \hspace{0.5em} \forall \hspace{0.5em} x \in (\ell, 1)$; of course $v(x, \ell) \geq 0 \hspace{0.5em} \forall \hspace{0.5em} x \in [0, \ell]$, so in fact $v(x, \ell) \geq 0$ everywhere in $[0, 1]$ and Bayes-plausibility is satisfied.
	
	The case where $L$ intersects $F$ twice in $(0, 1]$ will form the bulk of the proof. In particular, let $q_1$ be the smallest $q \in (0, 1]$ such that $\beta(\ell) \, q + y = F(q)$, and let $q_2$ be the largest.\footnote{Clearly, given the shape of $F$, $F(q) > L(q)$ in the interval $(q_1, q_2)$.} By the Implicit Function Theorem, since the function $\beta(\ell) \, q + y - F(q)$ is continuously differentiable in all variables, I can write $q_1$ and $q_2$ as continuous functions of $\ell$. Note that because of this definition, $q_1$ and $q_2$ are both well-defined (and satisfy $q_1 = q_2$) if $\beta(\ell) \, q + y$ is tangent to $F$, as well as for all smaller values of $\ell$. I now address two-intersection DTUs by focusing on the cases $q_1(\ell) > \ell$ and $q_1(\ell) \leq \ell$.
	
	If $q_1(\ell) > \ell$, then $G^{\beta(\ell)}_y(q) < F(q)$ for $q \in (0, \ell) \cup (q_1, q_2)$, but $G^\ell_y(q) > F(q)$ for $q \in [\ell, q_1) \cup (q_2, 1)$ (there is equality at $q \in \left\{0, q_1, q_2, 1\right\}$). Therefore if
	\begin{equation}
		\label{eq:intCond}
		v(q_1(\ell), \ell) = \int_0^{q_1(\ell)} F(q) \, dq - \int_0^{q_1(\ell)} G^\ell_y(q) \, dq \geq 0
	\end{equation}
	then $v(q) \geq 0 \hspace{0.5em} \forall \hspace{0.5em} q \in [0, 1]$ and Bayes-plausubility is satisfied. Given the increasing and decreasing behavior of $v(x, \ell)$, it is clear that 
	\begin{equation*}
		v(q_1(\ell), \ell) = \min_{q \in (0, 1)} v(q).
	\end{equation*}
	Therefore if a DTU violates Bayes-plausibility, it must be because $v(x, \ell) < 0$ for some $x \in (0, 1)$, which in turn implies that $v(q_1(\ell), \ell) < 0$. Thus when $q_1(\ell) > \ell$, Equation (\ref{eq:intCond}) is a necessary and sufficient condition for a DTU to be Bayes-plausible. Furthermore, if the inequality is strict for some $\ell$, then because $v(q_1(\ell), \ell)$ is continuous in $\ell$, it is also strict for $\ell - \varepsilon$.
	
	To close out the case $q_1(\ell) < \ell$, I show that either $q_1 < q_2 < 1$ or $\beta(\ell^\text{min}_y) \, q + y$ does not intersect $F$ twice. To see why, note that if $q_2(\ell) = 1$ then either $y = 1 - 2\pi$ and $\ell = 0$, or $y \in (1 - 2\pi, 1)$ and $\ell = \ell^\text{min}_y$. In the former case, any the DTU is actually a UTU, and any UTU intersects $F$ twice: otherwise it lies weakly above $F$ on the interval $[0, 1]$ and strictly above $F$ on some measurable subset of $[0, 1]$, and could not have the same mean as $F$, contradicting the construction of UTUs. Thus $q_1(\ell) < 1$, $v(x, \ell)$ is strictly increasing in $(q_1, q_2)$ and is negative at $x = q_1$, and $G^{\beta(\ell)}_y$ is not Bayes-plausible. In the latter case, if $\beta(\ell^\text{min}_y) \, q + y$ intersects $F$ twice, then the same argument applies and $G^\ell_y$ is not Bayes-plausible.\\
	
	If $q_1(\ell) \leq \ell$, then $G^{\beta(\ell)}_y$ satisfies Bayes-plausibility. It must be that $G_\ell(q) < F(q) \hspace{0.5em} \forall \hspace{0.5em} q \in (0, q_1) \cup (q_1, q_2)$, with equality at $q_1$ only if $q_1(\ell) = \ell$. Then $v(x, \ell) > 0$ on $(0, q_2)$, and since $v(x, \ell)$ is strictly decreasing on $(q_2, 1)$ with $v(1, \ell) = 0$, it must be that $v(q) > 0 \hspace{0.5em} \forall \hspace{0.5em} q \in (0, 1)$. However, I now prove that if $y \leq 1 - 2\pi$, then $G^{\beta(\ell)}_y$ cannot have minimal slope among all Bayes-plausible DTUs with intercept $y$. Towards proving this claim, I first show that as $\ell \rightarrow 0$, it cannot be that $q_1(\ell) \leq \ell$. Assume that for some $\ell_i$, $L$ intersects $F$ twice (so that $q_1$ and $q_2$ are distinct and well-defined) and $q_1(\ell_i) \leq \ell_i$. Then, for $\ell \in [0, \ell_i]$, the function $\beta(\ell) \, q + y$ will intersect $F$ twice. If $y > 0$, then because $F(0) = 0$ there is $\varepsilon > 0$ so that $\beta(\ell) \, q + y$ lies strictly above $F$ in $[0, \varepsilon)$ for any valid choice of $\ell$; thus $q_1(\ell) > \varepsilon$. If instead $y = 0$, then because $f(0) < 1 - 2\pi$, it must be that for any $\ell$, there is $\varepsilon > 0$ small enough that $F(\varepsilon) < (1 - 2\pi) \, \varepsilon \leq \beta(\ell) \, \varepsilon$ by convexity of $F$. Thus it is again true that $q_1(\ell) > \varepsilon$. In either case, taking $\ell < \varepsilon$\footnote{Of course, this choice may not be valid for $y > 1 - 2\pi$, since the lower bound on the set of valid $\ell$ is strictly above $\ell = 0$; if so, I cannot rule out that $q_1(\ell) \leq \ell$ for the minimum permissible $\ell$.} ensures that $\ell < q_1(\ell)$. To complete the proof, note that $\beta(\ell) \, \ell + y \leq  F(\ell)$ is a necessary condition for $q_1(\ell) \leq \ell$. Since $\beta(\ell) \, \ell + y$ is continuous in $\ell$, I can use the result above about $\ell \rightarrow 0$ to apply the Intermediate Value Theorem and find a value of $\ell \in (0, \pi]$ where $\beta(\ell) \, \ell + y = F(\ell)$ but $\ell - \varepsilon < q_1(\ell - \varepsilon)$ for any $\varepsilon > 0$ sufficiently small. Furthermore, $G^{\beta(\ell - \varepsilon)}_y$ is Bayes-plausible for $\varepsilon$ sufficiently small. When $q_1(\ell) = \ell$, it must be that $v(q_1(\ell), \ell) > 0$ since $G^\ell_y(q) < F(q) \hspace{0.5em} \forall \hspace{0.5em} q \in (0, \ell)$. By continuity of $v(q_1(\ell), \ell)$ in $\ell$, it must be that $v(q_1(\ell - \varepsilon), \ell - \varepsilon) > 0$ if $\varepsilon$ is sufficiently small. Since $\ell - \varepsilon < q_1(\ell - \varepsilon)$, Equation (\ref{eq:intCond}) is a necessary and sufficient condition for Bayes-plausibility of $G^{\beta(\ell - \varepsilon)}_y$, and therefore $G^{\beta(\ell - \varepsilon)}_y$ is Bayes-plausible and has a smaller slope than $G^\ell_y$.
	
	Having established sufficient conditions for when Bayes-plausibility is satisfied, I can now use them to obtain the desired characterization of the slope-minimizing lower truncation length $\ell^{sm}_y$. I begin with the case $y \in [0, 1 - 2\pi]$ and show that $\ell^{sm}_y$ satisfies $v(q_1(\ell^{sm}_y), \ell^{sm}_y) = 0$. When $y \in [0, 1 - 2\pi]$, the lowest permissible slope for a DTU is $(1 - y)^2/(2\pi)$, the slope of the UTU with intercept $y$. Therefore the line $L(q) = q \, (1 - y)^2/(2\pi) + y$ must intersect $F$ twice in $(0, 1]$. Furthermore, the line $L(q) = q \, (1 - y)/\pi + y$ corresponds to the maximum permissible slope for a DTU, and thus must lie above $F$ for the mean of that DTU to equal the mean of $F$. Therefore by continuity of $\beta(\ell)$ in $\ell$ and continuity of $f$, there exists a value $\ell_t \in (0, \pi)$ where the line $L(q) = \beta(\ell) \, q + y$ is tangent to $F$. The point of tangency must be interior, as $\beta(\ell_t) \cdot 1 + y = 1$ only if $\ell_t = 0$, in which case the line $\beta(\ell_t) \, q + y$ forms part of a UTU and (as argued above) cannot be tangent to $F$. Therefore, for $\varepsilon > 0$ sufficiently small the line $\beta(\ell_t - \varepsilon) \, q + y$ intersects $F$ twice, and both intersections are bounded strictly below 1. As argued when showing that $q_1(\ell) \leq \ell$ implies Bayes-plausibility of $G^{\beta(\ell)}_y$, the constraint in Equation (\ref{eq:intCond}) does not bind for $G^{\ell_t}_y$, so it does not bind for $G^{\ell_t - \varepsilon}_y$, and the latter DTU is therefore Bayes-plausible. Thus the $y$-optimal DTU $G^{sm}_y$ cannot be tangent to $F$ and must intersect $F$ twice in $(0, 1]$. Since $y \in [0, 1 - 2\pi]$, as shown for the case $q_1(\ell) \leq \ell$ it cannot be that $q_1(\ell^{sm}_y) \leq \ell^{sm}_y$. Therefore $q_1(\ell) > \ell$ and the necessary and sufficient condition for Bayes-plausibility in Equation (\ref{eq:intCond}) applies. To show that it holds with equality, consider the UTU corresponding to $\ell = 0$. It is not Bayes-plausible\footnote{Any UTU with $y = 0$ has an atom at 0 while $F$ does not. If $y = 0$, the restriction that $f(0) < 1/(2\pi)$ ensures that the UTU is not Bayes-plausible, since there is $\varepsilon > 0$ such that the UTU places more mass in the interval $[0, \varepsilon]$ than does $F$.} and intersects $F$ twice, so it must be that $v(q_1(0), 0) < 0$. Because $v(q_1(\ell), \ell)$ is a continuous function of $\ell$ that takes both positive and negative values for $\ell \in [0, \pi]$, the Intermediate Value Theorem implies that there is a well-defined minimum value of $\ell$, which I call $\ell_m$, for which $v(q_1(\ell_m), \ell_m) = 0$. Since $v(q_1(\ell), \ell) < 0$ for any $\ell < \ell_m$, and I have shown that $v(q_1(\ell^{sm}_y), \ell^{sm}_y) \geq 0$, it must therefore be that $\ell^{sm}_y = \ell_m$.
	
	To complete the proof of the lemma, I show that if $y \in (1 - 2\pi, 1)$, then either $\ell^{sm}_y = \ell^\text{min}_y$ or $v(q_1(\ell^{sm}_y), \ell^{sm}_y) = 0$. Assume that $\beta(\ell^\text{min}_y) \, q + y$ intersects $F$ twice; otherwise clearly $G^{\ell^\text{min}_y}_y$ is Bayes-plausible and $\ell^{sm}_y = \ell^\text{min}_y$. Assume also that the smallest $\ell$ for which $v(q_1(\ell), \ell) = 0$, which I label $\ell^0_y$, satisfies $\ell_y^0 > \ell^\text{min}_y$; otherwise clearly $G^{\ell^0_y}_y$ is both Bayes-plausible and slope-minimizing, so again $\ell^{sm}_y = \ell^\text{min}_y$ (if no $\ell$ satisfying $v(q_1(\ell), \ell) = 0$ exists, I let $\ell^0_y = \pi$, and the argument still holds). If $\ell^{sm}_y \in (\ell^\text{min}_y, \ell^0_y)$, then it must be that $\beta(\ell^{sm}_y) \, q + y$ intersects $F$ twice, because $\beta(\ell^0_y) \, q + y$ does. By the definition of $\ell^0_y$, $v(q_1(\ell^*_y), \ell^{sm}_y) \neq 0$. Clearly that expression cannot be strictly positive, or by continuity there would be $\varepsilon > 0$ small enough so that $\ell^{sm}_y - \varepsilon$ is both a valid choice of $\ell$ (i.e., greater than $\ell^\text{min}_y$) and generates a Bayes-plausible DTU. It must therefore be strictly negative, which means that $q_1(\ell^{sm}_y) \leq \ell^{sm}_y$; otherwise $G^{sm}_y$ would not be Bayes-plausible. But then the proof that $q_1(\ell) \leq \ell$ cannot occur for small $\ell$ implies that there is $\varepsilon > 0$ small enough so that $\ell^{sm}_y - \varepsilon > \ell^\text{min}_y$ and $G^{\ell^{sm}_y - \varepsilon}_y$ is Bayes-plausible, which contradicts the slope-minimizing property of $\ell^{sm}_y$ (the caveat for $y > 1 - 2\pi$ does not apply, since we have already covered and ruled out the case $\ell^{sm}_y = \ell^{\text{min}}_y$). Thus it cannot be true that $\ell^{sm}_y \in (\ell^\text{min}_y, \ell^0_y)$, so it must be that either $\ell^{sm}_y = \ell^\text{min}_y$ or $\ell^{sm}_y = \ell^0_y$; the latter implies the desired condition $v(q_1(\ell^{sm}_y), \ell^{sm}_y) = 0$.
\end{proof}

\subsection*{B4: Overall-Optimal DTUs}
\label{subsec:b4}
The simplified integral constraint in Lemma \ref{lm:contEll} can be used as a key step in deriving the continuity of $\ell^*_y$ in $y$, and thus in providing sufficient conditions for the existence of an overall-optimal DTU in Lemma \ref{lm:existOpt}. As an immediate corollary, though, it allows a characterization of the overall-optimal DTU when $r^*$ is small:
\begin{corollary}
	\label{cor:optLowR}
	Let $\beta(0)$ be the slope of the $0$-optimal double-truncated uniform distribution $G^{\beta(0)}_0$, and let $q_i(\beta(0), 0)$ be the smallest $q \in (0, 1)$ that satisfies $\beta(0) \, q = F(q)$.\footnote{The existence of $q_i(\beta(0), 0)$ is guaranteed by the proof of Lemma \ref{lm:contEll}.} If $r^* \leq q_i(\beta(0), 0)$, then $G^{\beta(0)}_0$ is uniquely optimal among all double-truncated uniform distributions.
\end{corollary}
\begin{proof}
	The proof is by contradiction, and resembles the geometric proof of Proposition \ref{prop:bin} for the binary-state setting. Fix $r^*$ and assume some other DTU $G$ does weakly better than $G^{\beta(0)}_0$ for Sender. It must therefore have a smaller slope than $G^*_0$: the intercept of $G$ is larger than that of $G^{\beta(0)}_0$, and $G$ must intersect the horizontal line $y = 1$ at a larger value of $q$ than $G^{\beta(0)}_0$ or the concavification of $G$ would be everywhere above that of $G^{\beta(0)}_0$. Because of its larger slope, $G^{\beta(0)}_0$ upper-bounds $G$ after $r^*$ (where $G$ lies weakly below $G^{\beta(0)}_0$) and thus 
	\begin{equation*}
		\begin{split}
			& \int_{r^*}^1 G(q) \, dq < \int_{r^*}^1 G^{\beta{0}}_0(q) \, dq 
			\\& \Rightarrow 
			\int_{q_i(\beta(0), 0)}^1 G(q) \, dq < \int_{q_i(\beta(0), 0)}^1 G^{\beta(0)}_0(q) \, dq
			\\& \Rightarrow \int_0^{q_i(\beta(0), 0)} G(q) \, dq > \int_0^{q_i(\beta(0), 0)} G^{\beta(0)}_0(q) \, dq = \int_0^{q_i(\beta(0), 0)} F(q) \, dq.
		\end{split}
	\end{equation*}
	The inequality in the first line is strict because $F(q_1(\beta(0), 0)) < 1$, so $r^*$ is not in the upper-truncated region of $G$ and there is some strict difference between $G^{\beta(0)}_0$ and $G$ captured in the integral. The first implication follows from the bound on $r^*$. The inequality in the third line is because all DTUs have equal means, so 
	\begin{equation*}
		\begin{split}
			1 - \pi & = \int_0^{1} G(q) \, dq = \int_0^{q_i(\beta(0),0)} G(q) \, dq + \int_{q_i(\beta(0),0)}^1 G(q) \, dq 
			\\& = \int_0^1 G^{\beta(0)}_0(q) \, dq = \int_0^{q_i(\beta(0),0)} G^{\beta(0)}_0(q) \, dq + \int_{q_i(\beta(0),0)}^1 G^{\beta(0)}_0(q) \, dq.
		\end{split}
	\end{equation*}
	The equality in the third line is by Lemma \ref{lm:contEll}, since by Lemma \ref{lm:minSlope} the DTU $G^{\beta(0)}_0$ has minimal slope among Bayes-plausible DTUs with intercept $0$.
\end{proof}

Using the characterization of Lemma \ref{lm:contEll}, I now prove a sufficient condition on $F$ for $U_{r^*}$, the set of utilities attained by $y$-optimal DTUs, to be compact, and thus for Sender to have a well-defined overall-optimal DTU:
\begin{lemma}
	\label{lm:existOpt}
	Let $r^* \in [\pi, 1]$ and $f(1) > 0$. Then Sender's maximum utility over all double-truncated uniform distributions is well-defined, and is attained by a double-truncated uniform distribution $G^*$.
\end{lemma}
\begin{proof}
	I first show that the optimal lower truncation length $\ell^*_y$ is continuous in $y$ at any $y \in [0, 1)$. Given the restriction on $r^*$, Sender's utility from a $y$-optimal DTU $G^{\beta(y)}_y$ is given by $1 - G^{\beta(y)}_y = 1 - (\beta(y), y) \, r^* + y)$. Thus continuity of $\ell^*_y$ in $y$ ensures that Sender's maximum utility over DTUs with intercept $y$ is continuous in $y$. I can then provide sufficient conditions for the intercept of a potential overall-optimal DTU to lie in a compact set. The continuity condition implies that $U_{r^*}$ is compact, so that it contains its closure. Therefore there is some DTU that attains Sender's supremum utility over all DTUs.
	
	To show continuity, I first work with $y \in [0, 1 - 2\pi)$, where the argument is most straightforward. Since in that range $v(q_1(\ell^*_y, y), \ell^*_y) = 0$ by Lemma \ref{lm:contEll}, and the proof of that lemma shows that $\ell^*_y$ is the minimal $\ell$ where the property holds, I can apply the Implicit Function Theorem to write $\ell^*_y$ as a continuous function of $y$.
	
	When $y \in (1 - 2\pi, 1)$, then Lemma \ref{lm:contEll} implies that either $v(q_1(\ell^*_y, y), \ell^*_y) = 0$ or $\ell^*_y = \ell^\text{min}_y$. In particular, $\ell^*_y$ is either the minimum permissible $\ell$ or, if that choice does not deliver a Bayes-plausible DTU, the minimum $\ell$ satisfying $v(q_1(\ell, y), \ell) = 0$. Because both $\ell^\text{min}_y$ and the minimal $\ell$ satisfying $v(q_1(\ell, y), \ell) = 0$ are continuous in $y$, the minimum over those two choices is also continuous in $y$. Thus $\ell^*_y$ is continuous in $y$ for $y \in (1 - 2\pi, 1)$. 
	
	All that remains is to show that $\ell^*_y$ is continuous in $y$ at $y = 1 - 2\pi$. The continuity of $\ell^\text{min}_y$ in $y$ ensures that the function
	\begin{equation*}
		u(y) = \beta(\ell^\text{min}_y) \, \ell^\text{min}_y + y - F(\ell^\text{min}_y)
	\end{equation*}
	is also continuous in $y$. Because $u(y) > 0$ for any $y \in [0, 1 - 2\pi]$, as shown in the proof of why $q_1(\ell) > \ell$ for small enough $\ell$, it must be that for $\delta > 0$ sufficiently small, $u(y') > 0$ for any $y' \in (1 - 2\pi, 1 - 2\pi + \delta)$. Since the line $\beta(0, 1 - 2\pi) \, q + (1 - 2\pi)$ intersects $F$ twice in $(0, 1]$, it must therefore be that for $\delta > 0$ sufficiently small and $y' \in (1 - 2\pi, 1 - 2\pi + \delta)$, so do the lines $\beta(\ell^\text{min}_{y'}, 1 - 2\pi) \, q + (1 - 2\pi)$, $\beta(\ell^\text{min}_{y'}, y') \, q + (1 - 2\pi)$, and $\beta(\ell^\text{min}_{y'}, y') \, q + y'$. Because the last intersects $F$ twice in $(0, 1]$, and both intersections occur at values $q > \ell^\text{min}_{y'}$, the proof of Lemma \ref{lm:contEll} shows that $v(q_1(\ell^*_{y'}, 0), \ell^*_{y'}) = 0$ and $\ell^*_{y'}$ is the minimal value of $\ell$ such that this property holds. Therefore, by the continuity of the minimal value of $\ell$ satisfying this equation, $\ell^*$ is continuous in $y$ at $y = 1 - 2\pi$.
	
	Having shown continuity of $\ell^*_y$ in $y$, I use the second part of the lemma statement to show that the set of possibly overall-optimal DTU intercepts is compact. Given that $f(1) > 0$, there must be $\bar{y} \in (0, 1)$ such that $1 - \bar{y} > f(1)$. Then for any intercept $y \geq \bar{y}$, the DTU with minimal permissible slope lies above $F$ on $(0, 1)$, and is therefore Bayes-plausible. Since any DTU with intercept $y > \bar{y}$ surely lies above the slope-minimal DTU with intercept $\bar{y}$ for all $q \in [\pi, 1]$, no DTU with intercept in $(\bar{y}, 1)$ can be optimal among all DTUs. Thus the intercept of the overall-optimal DTU lies in the compact set $[0, \bar{y}]$.
\end{proof}

Note that only the last step of the proof relies on $f(1) > 0$; if this condition is violated, then (as in the statement of Proposition \ref{prop:contLarge} in the text) it may be that no DTU attains Sender's supremum utility, but there exists a limiting sequence of DTUs converging to that value so no distribution delivers strictly higher utility than all DTUs.

\subsection*{B5: Optimal Posterior Distributions}
\label{subsec:b5}
Having established properties of overall-optimal DTUs, I can now jointly prove the optimal distribution portions of Propositions \ref{prop:contSmall} and \ref{prop:contLarge}:
\begin{proof}
	Let $H$ be a candidate optimal distribution of posterior means. I approximate $\bar{H}$, the concavification of $H$, by a tangent at $r^*$, which I call $L(q)$; let $L(0) = y_L \in [0, 1)$ be its intercept. Consider the $y_L$-optimal DTU $G^{\beta(y_L)}_{y_L}$. In order for $H$ to do at least as well for Sender as $G^{\beta(y)}_{y_L}$, by Lemma \ref{lm:cav} it must be that
	\begin{equation*}
		1 - H(r^*) \geq 1- \bar{H}(r^*) \geq 1 - G^{\beta(y_L)}_{y_L}(r^*).
	\end{equation*}
	Thus $L$ must have a weakly smaller slope than $G^{\beta(y_L)}_{y_L}$; otherwise $L(r^*) > G^\beta(y_L)_{y_L}(r^*)$ and the above inequality is violated. 
	
	If $y_L \in [0, 1 - 2\pi]$, then for any slope $\beta \in ((1 - y_L)^2/(2\pi), \beta(y_L)]$ there is a DTU with that slope and intercept $y_L$. If instead $y_L \in (1 - 2\pi, 1)$, then for any slope $\beta \in [1 - y_L, \beta(y_L)]$ there is a DTU with that slope and intercept $y_L$. In the first case, the slope of $L$ cannot lie below that interval or it would have a weakly smaller slope than the UTU with intercept $y_L$; then the argument of Proposition \ref{prop:bin} applies and $H$ is not Bayes-plausible. In the second case, $L$ must have a slope weakly greater than the lowest-slope DTU with intercept $y_L$, or it would fail to pass through $(1, 1)$, and therefore so would $\bar{H}$ and $H$. Thus there is a DTU $G^L$ with the same slope as $L$.
	
	Let $r^* \in (0, q_i(\beta(0), 0)]$. By Corollary \ref{cor:optLowR}, if $G^L \neq G^{\beta(0)}_0$, then because $G^L(r^*) \leq G^{\beta(0)}_0(r^*)$, $G_L$ is not Bayes-plausible. If instead $r^* \in [\pi, 1)$ and $G^L$ has a strictly smaller slope than $G^{\beta(y_L)}_{y_L}$, then by Lemma \ref{lm:minSlope}, $G^L$ is not Bayes-plausible. 
	
	In either case, given that $G^L$ violates Bayes-plausibility, $H$ must violate it as well. Because $G^L$ upper-bounds $\bar{H}$ beyond $\ell$, it must be that
	\begin{equation*}
		\int_q^1 H(t) dt \leq \int_q^1 \bar{H}(t) dt \leq \int_q^1 G^L(t) dt
	\end{equation*}
	for any $q \in [\ell, 1]$. Since $G^L$ violates Bayes-plausibility, there is some $q_v \in [0, 1]$ where
	$$\int_0^{q_v} G^L(t) dt > \int_0^{q_v} F(t) dt,$$
	and since the left-hand side equals 0 for any $q \in [0, \ell)$, it must be that $q_v \in [\ell, 1]$. Then because $H$ and $G^L$ have the same mean,
	\begin{equation*}
		\begin{split}
			& \int_0^1 H(t) dt = \int_0^1 G^L(t) dt = 1 - \pi
			\\& \Rightarrow \int_0^{q_v} H(t) dt + \int_{q_v}^1 H(t) dt = \int_0^{q_v} H(t) dt + \int_{q_v}^1 H(t) dt
			\\& \Rightarrow \int_0^{q_v} H(t) dt \geq \int_0^{q_v} G^L(t) dt > \int_0^{q_v} F(t) dt,
		\end{split}
	\end{equation*}
	where the third line follows from the earlier upper bound on the integral of $H$. Therefore $H$ violates Bayes-plausibility and is not a valid distribution.
	
	If $G^L$ has the same slope as $G^{\beta(y_L)}_{y_L}$, then by construction $H$ gives Sender the same utility as $G^{\beta(y_L)}_{y_L}$. Thus if there is a DTU that delivers Sender a strictly higher utility than $G^{\beta(y_L)}_{y_L}$, then clearly $H$ is not optimal overall. If there is no such DTU, then $G^{\beta(y_L)}_{y_L}$ is optimal among all DTUs and $H$ also attains Sender's maxmin utility.
\end{proof}

Finally, I prove the unique concavification portion of Proposition \ref{prop:contSmall}:
\begin{proof}
	By the proof for optimal distributions above, the slope of $\bar{H}$ at $r^*$ equals that of $G^{\beta(0)}_0$. Because $\bar{G}^{\beta(0)}_0$ does not have a kink at $\ell^*_0$, it upper-bounds $H$ on the whole interval $[0, 1]$ instead of just on $[\ell, 1]$ as in that proof. If $H < G^{\beta(0)}_0$ on any measurable subset of $[r^*, 1]$ the proof of Corollary \ref{cor:optLowR} shows that $H$ violates the Bayes-plausibility integral constraint at $q_i(\beta(0), 0)$. 
	
	If $\ell^*_0 \leq r^*$, it is therefore true that $\bar{H} = G^{\beta(0)}_0 = \bar{G}^{\beta(0)}_0$ on $[r^*, 1]$. Furthermore, $\bar{G}^{\beta(0)}_0$ upper-bounds $\bar{H}$ on $[0, r^*]$ and $\bar{G}^{\beta(0)}_0(0) = \bar{H}(0) = 0$. Because $\bar{G}^{\beta(0)}_0$ is linear on $[0, r^*]$ (i.e., it has no kink at $\ell^*_0$) there is no smaller concave function that takes the same values at $q = 0$ and $q = r^*$; thus $\bar{H} = \bar{G}^{\beta(0)}_0$ on $[0, r^*]$ as well.
	
	If instead $\ell^*_0 > r^*$, then it is now the case that $\bar{H} = G^{\beta(0)}_0 = \bar{G}^{\beta(0)}_0$ on $[\ell^*_0, 0]$, since that is the range where the latter equality holds. However, the upper-bounding relationship still holds on $[0, \ell^*_0]$, and thus the argument above still applies and $\bar{H} = \bar{G}^{\beta(0)}_0$ on $[0, \ell^*_0]$.
\end{proof}

\subsection*{B6: The Finite-State Case}
\label{subsec:b6}
Note that the proofs and results of Lemmas \ref{lm:intOpt} and \ref{lm:minSlope} go through unchanged. Thus I can in fact easily prove an analogue for Proposition \ref{prop:contLarge} by following the proof in Appendix \hyperref[subsec:b5]{B5}. In particular, a candidate optimal distribution can be approximated by a DTU $G^L$. If $G^L$ has a strictly smaller slope than the $y$-optimal DTU with the same intercept $y$, then by assuming $r^* \in [\pi, 1]$ and following Lemma \ref{lm:minSlope}, it cannot be Bayes-plausible. Given that $G^L$ is not Bayes-plausible, neither is the candidate optimal distribution. Since there is a $y$-optimal DTU for any possible $y$, no distribution can give Sender strictly higher utility than all DTUs.

To obtain a tighter characterization of which DTU is optimal in this setting, I can prove an analogue of Lemma \ref{lm:contEll}, showing where the integral constraint binds for the $y$-optimal DTU:
\begin{lemma}
	\label{lm:yOptimalFiniteN}
	Let $y = 0$ or $r^* \in [\pi, 1]$ so that the slope-minimizing Bayes-plausible DTU is also Sender's optimal DTU. Fix the value of $y$, and let $q_i$ be the location of the $i$th interior atom of the prior $F$. Then either the optimal lower truncation length $\ell^*_y$ is equal to the minimum lower truncation length $\ell^\text{min}_y$, or for at least one $i \in \left\{1,...,N - 2\right\}$ it is true that
	\begin{equation*}
		\lim_{x \rightarrow q_i^-} \int_0^x F(q) \, dq - \int_0^x G_y^\beta = 0.
	\end{equation*}
\end{lemma}
\begin{proof}
	 For completeness, define $q_0 = 0$.
	
	The proof is algorithmic; the algorithm for finding the optimal DTU is as follows.
	\begin{enumerate}
		\item Initialize $\beta$ as the minimal feasible slope for a DTU with intecept $y$.
		\item For $i \in 1,...,N-2$:
		\begin{enumerate}
			\item Check whether the integral constraint is satisfied at the left limit of $q_i$. That is, whether
			\begin{equation*}
				\lim_{x \rightarrow q_i^-} \int_0^x F(q) \, dq - \int_0^x G_y^\beta \geq 0.
			\end{equation*}
			\item If the constraint is satisfied, exit.
			\item Else, increase $\beta$ until the constraint binds exactly.\footnote{By the proof of Lemma \ref{lm:intOpt}, the difference of integrals is continuous and monotonically increasing in $\beta$, so there will be exactly one value where the constraint binds.}
		\end{enumerate}
		\item Return $\beta = \beta(y)$, the $y$-optimal slope.
	\end{enumerate}
	In step (b), if the integral constraint is satisfied at the left limit of $q_i$, it must be satisfied everywhere in $[q_{i - 1}, q_i)$, since in that interval $F$ is constant but $G_y^\beta$ is weakly decreasing. The algorithm first finds the minimal value $\beta$ where $G_y^\beta$ satisfies the integral constraint in $[0, q_1)$, then proceeds across subsequent intervals, increasing $\beta$ if necessary to ensure the integral constraint is satisfied. Finally, the integral constraint is automatically satisfied in $[q_{N - 2}, 1]$ because $G_y^\beta$ has the appropriate mean.
	
	Thus either the initial value of $\beta$ satisfies the integral constraint for all $i$, in which case that minimal feasible slope is $y$-optimal, or the constraint is satisfied for at least one $i$, giving the result in the lemma.
\end{proof}
As in the continuous-state case, the integral constraint binds at a finite and possibly empty set of points for each $y$-optimal DTU. In the continuous-state case, this set was guaranteed to be nonempty for all $y \in [0, 1 - 2\pi]$; however if $F(0) > 0$ in this finite-support setting, the set may be empty for even $G_y^\beta(0)$, the $0$-optimal DTU. In the case where there does exist a minimal $q_i$\textemdash call it $q_\text{min}$\textemdash where the integral constraint binds for $G_0^{\beta(0)}$ (a fact which depends on the specification of $F$), there is a natural analogue of Corollary \ref{cor:optLowR}: if $r^* \in [0, q_\text{min}]$, then $G_0^{\beta(0)}$ is Sender's overall-optimal DTU. The proof exactly parallels that of the original corollary. Indeed, the analogue of Proposition \ref{prop:contSmall} also follows, since the proof of optimality in Appendix \hyperref[subsec:b5]{B5} then goes through in the same way.

\newpage

\section*{Appendix C: A Numerical Approach to the Continuous-State Setting}
\label{sec:c}
\subsection*{C1: Summary of Numerical Results}
\label{subsec:c1}
Propositions \ref{prop:contSmall} and \ref{prop:contLarge} leave open the optimal distribution of posterior means when the mean Receiver type $r^*$ lies in $\pi \in (q_i(\beta(0), 0), \pi)$. While the value $q_i(\beta(0), 0)$ is well-defined for a given prior distribution, a closed-form solution may not exist. However, fixing a prior distribution, I can use a two-step solution algorithm to numerically compute $q_i(\beta(0), 0)$ and show qualitatively how the size of the intermediate interval changes with various properties of the prior distribution.

Informally, given a prior $F$ with mean $\pi$, the first step is to find the $0$-optimal slope $\beta(0)$. To check Bayes-plausibility, I use the simplified integral constraint from Lemma \ref{lm:contEll}. Starting with the minimum $\beta$, I increase $\beta$ only if the constraint is violated and stop when it binds. The second step checks for intersections between the $0$-optimal DTU and $F$; by definition the smallest interior intersection is $q_i(\beta(0), 0)$. A full formal description is in Appendix \hyperref[subsec:c2]{C2}. 

I briefly discuss some intuition for the results of the numerical computation below. A reader interested in further detail may refer to Appendix \hyperref[subsec:c3]{C3} for thorough figures showing the output of the algorithm at various parameter values, or to Appendix \hyperref[subsec:c4]{C4} for a detailed exposition of those figures. Throughout this section $\mu$ and $\sigma$ refer to the mean and standard deviation of the generating normal distribution $N(\mu, \sigma^2)$ while $\pi$ refers to the mean of the prior $F$, i.e., $N(\mu, \sigma^2)$ truncated in $[0, 1]$. 

For fixed $\mu$, the $0$-optimal slope $\beta(0)$ and the $0$-optimal lower truncation length $\ell(0)$ are decreasing in $\sigma$. A smaller slope is better for Sender, but may be ruled out by the integral constraint; increasing $\sigma$ means the prior cdf increases less steeply, so the integral constraint allows Sender's chosen distribution to increase less steeply as well. 

For $\mu < 1/2$ (the midpoint of the truncation interval), the key value $q_i(\beta(0), 0)$ has an inverse-U-shaped relationship with $\sigma$, while for $\mu > 1/2$ it is increasing in $\sigma$. Increasing $\sigma$ lowers the slope of the $0$-optimal DTU, which would decrease $q_i(\beta(0), 0)$ if the shape of the prior were unchanged. However, holding fixed a DTU's slope, increasing $\sigma$ spreads out the prior mass and increases $q_i(\beta(0), 0)$. For small $\mu$, either of these effects can dominate. For large $\mu$, most of the prior mass is far enough away from the origin that the second effect dominates.

The prior mean $\pi$ always lies strictly below $q_i(\beta(0), 0)$ for $\sigma$ small enough, and increases with $\sigma$ when $\mu < 1/2$ but decreases with $\sigma$ when $\mu > 1/2$. This behavior is a known property of the truncated normal distribution; in my setting, it implies that there is no gap between Proposition \ref{prop:contSmall} and Proposition \ref{prop:contLarge} whenever $\mu > 1/2$. When $\mu < 1/2$, it implies that eventually $q_i(\beta(0), 0) < \pi$, producing a gap between the results.

\subsection*{C2: A Detailed Algorithm for Computing $q_i(\beta(0), 0)$}
\label{subsec:c2}
In this section I describe in detail the algorithm for computing $q_i(\beta(0), 0)$, as well as some notes on its key steps and the details of implementation.

\begin{enumerate}[label={\arabic*.}]
	\item Compute the $0$-optimal slope $\beta(0)$:
	\begin{enumerate}
		\item Following Appendix \hyperref[subsec:b1]{B1}, the minimal DTU slope when the intercept is $y = 0$ is $\beta_\text{min} = \max \left\{1/(2\pi), 1\right\}$. Initialize the slope $\beta$ at $\beta_\text{min}$.
		\item Define the function 
		\begin{equation*}
			\ell(\beta) = \frac{\sqrt{2\pi \beta - 1}}{\beta},
		\end{equation*}
		the lower truncation length of a DTU with slope $\beta$ and intercept $y = 0$ (as described in Appendix \hyperref[subsec:b1]{B1}).
		\item Given $\beta$, define the function
		\begin{equation*}
			v(q) = \int_0^q F(x) \, dx - \int_0^x G^\beta_{\ell(\beta)} (x) \, dx
		\end{equation*}
		\item Compute $v^* = \min_{q \in (0, 1)} v(q)$ using standard numerical optimization. 
		\item If $v^* < 0$, increase $\beta$ and return to 2(b). Else, return $\beta(0) = \beta$.
	\end{enumerate}
	 \item Find the smallest interior intersection $q_i(\beta(0), 0)$:
	 \begin{enumerate}
	 	\item Given $\beta(0)$, solve $\beta(0) \, q = F(q)$ using a standard numerical solver.
	 	\item Discard the solution $q = 0$; return the smallest remaining solution as the value of $q_i(\beta(0), 0)$. 
	 \end{enumerate}
\end{enumerate}
Step 1(e) is guaranteed to terminate because for $\beta = \beta_\text{max} = 1/\pi$ the DTU will lie weakly above $F$ everywhere after the lower truncation region $[0, \ell(\beta)]$; following the proof of Lemma \ref{lm:contEll}, if this relationship holds then $G^\beta_{\ell(\beta)}$ is Bayes-plausible. Step 2(a) is guaranteed to find an non-zero solution, a result of the same proof.

In practice, step 2 was computed by defining the minimum and maximum values of $\beta$ given $\pi$ (as described in Appendix \hyperref[subsec:b1]{B1}) and considering a 1,000-point grid over that interval. Then, since an analytic expression for the truncated normal does not exist, the integral in step 2(c) was computed numerically over a 10,000-point grid covering $[0, 1]$. The procedure was repeated for prior distributions constructed using $\mu = (0.01,0.02,...0.99)$ and $\sigma = (0.01,0.015,...0.50)$ for the generating normal distribution. Larger values of $\sigma^2$ produced errors; even the maximum precision of the numerical integral was not high enough to accurately complete steps 1(d) and 1(e).

\subsection*{C3: Plotting $0$-Optimal DTUs}
\label{subsec:c3}
Below I include detailed plots of the 0-optimal slope $\beta(0)$, the 0-optimal lower truncation length $\ell(0)$, and the value $q_i(\beta(0), 0)$ used in Proposition \ref{prop:contSmall}. In each figure, the prior is a truncation to $[0, 1]$ of a $N(\mu, \sigma^2)$ distribution. Each panel fixes a value $\mu \in \left\{0.1, 0.25, 0.4, 0.55, 0.7, 0.85\right\}$ and shows the relevant values as functions of the standard deviation $\sigma$.

\clearpage
\begin{figure}[t!]
	\centering
	\includegraphics[width=\hsize]{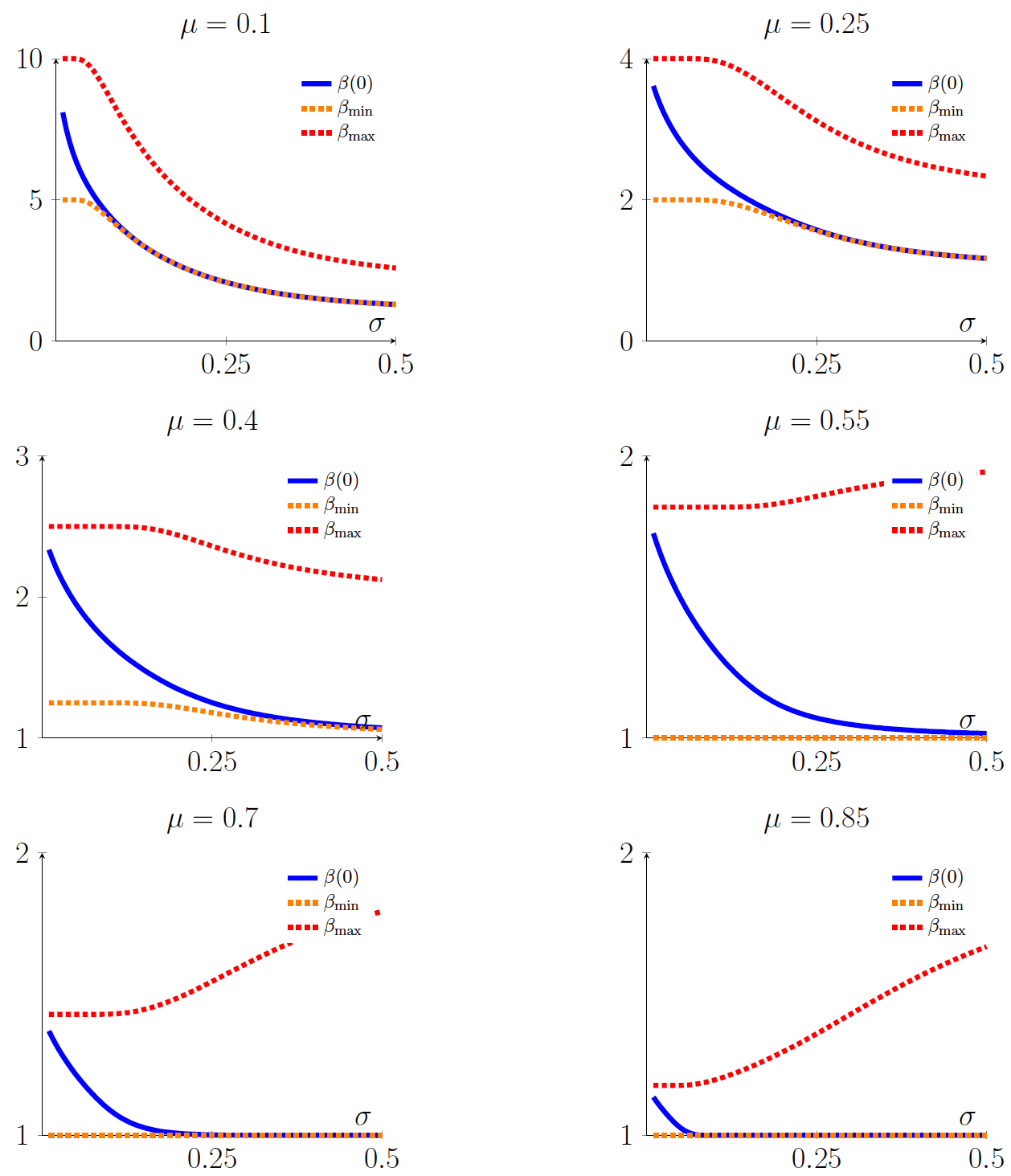}
	\caption{For fixed $\mu$, values of $\beta(0)$ and the upper and lower bounds on $\beta$ as functions of $\sigma$.}
	\label{fig:betaVals}
\end{figure}

\clearpage
\begin{figure}[t!]
	\centering
	\includegraphics[width=\hsize]{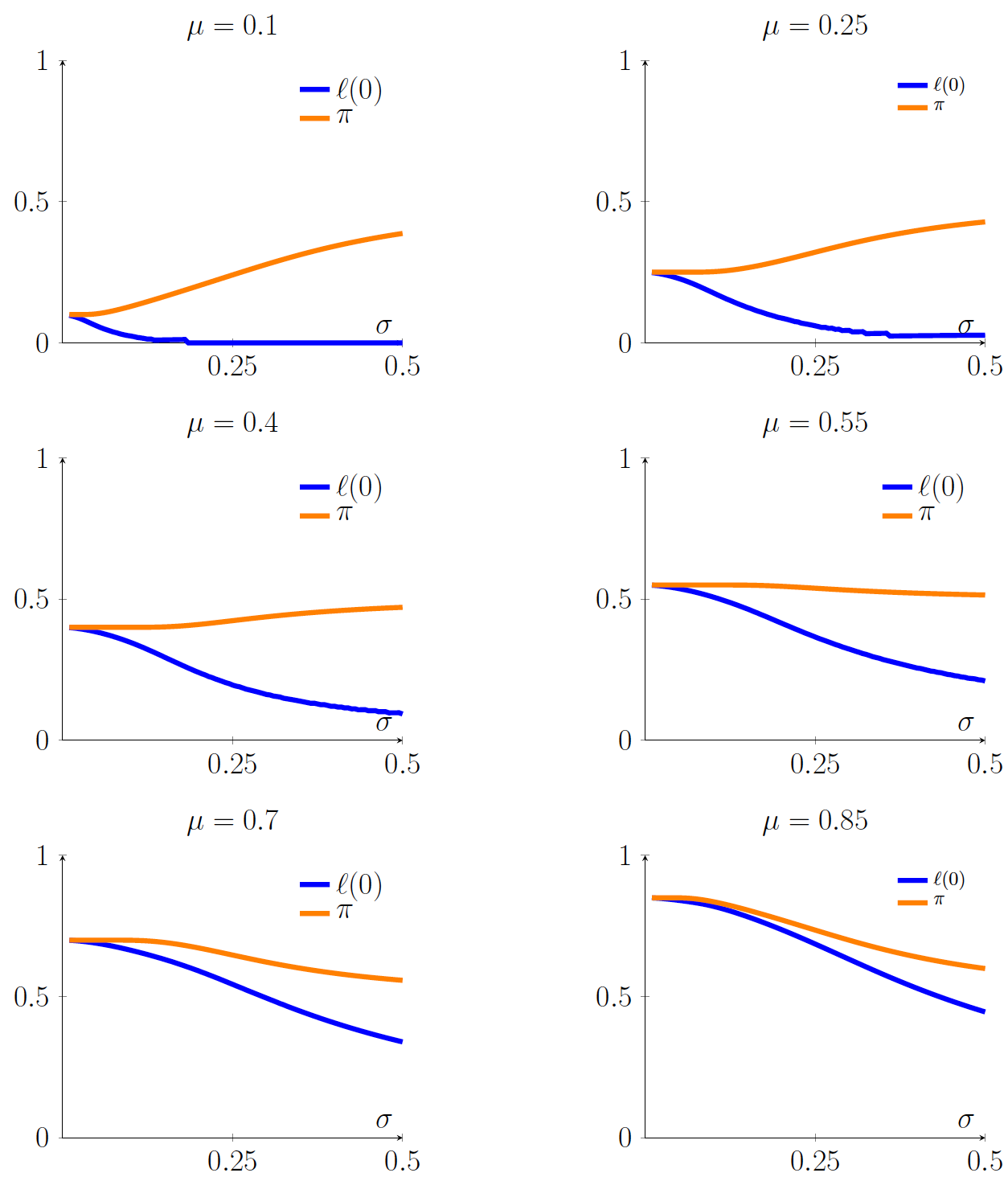}
	\caption{For fixed $\mu$, values of $\ell(0)$ and the upper bound $\ell_\text{max} = \pi$ as functions of $\sigma$.}
	\label{fig:ellVals}
\end{figure}

\clearpage
\begin{figure}[t!]
	\centering
	\includegraphics[width=\hsize]{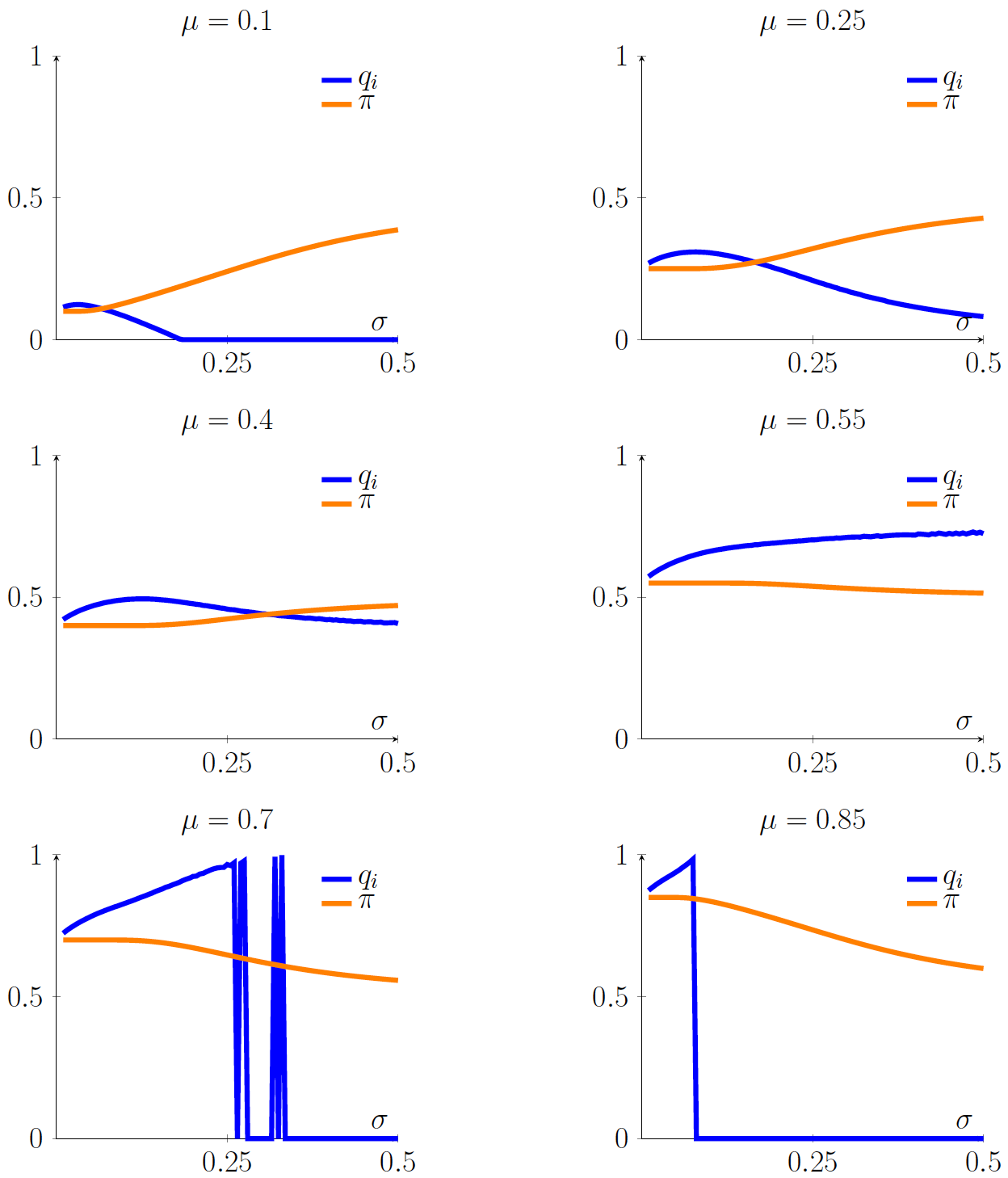}
	\caption{For fixed $\mu$, values of $q_i(\beta(0), 0)$ and $\pi$ as functions of $\sigma$.}
	\label{fig:qVals}
\end{figure}

\clearpage
\subsection*{C4: Discussion of Numerical Results}
\label{subsec:c4}
Figure \ref{fig:betaVals} shows clear patterns in the $0$-optimal choice of $\beta$, $\beta(0)$, across prior parameter values. Fixing $\mu$, $\beta(0)$ is monotonically decreasing in $\sigma$ with a roughly exponential shape. As $\mu$ increases, the range of $\beta(0)$ decreases. Note that for $\sigma$ large, the assumption $f'(0) < 1 - 2\pi$ used in Section \ref{sec:cont} to rule out the binary-state solution may be violated. In this case, Proposition \ref{prop:contSmall} does not apply as $\beta(0) = \beta_\text{min}$ may be a feasible solution. This case occurs in panels 1 and 2 of the figure. A similar phenomenon occurs when $\mu$ is large, so that the prior $F$ is concave (as described in Section \ref{subsec:simpCont}), and appears in panels 5 and 6 of the figure.

The $0$-optimal lower truncation length $\ell(0)$ behaves as expected given the results for $\beta(0)$, and is shown in Figure \ref{fig:ellVals}. For fixed $\mu$, $\ell(0)$ is monotonically decreasing and appears to have a reverse-S shape (concave and then convex). Increasing $\mu$ shifts the curve up and flattens it. As with $\beta(0)$, the results in the first two panels show some instability at high values of $\sigma$ resulting from violations of the assumption $f'(0) < 1 - 2\pi$. In panels 5 and 6, where $\beta(0) = \beta_\text{min}$ at high values of $\sigma$, the lower truncation length tracks directly with the prior mean $\pi$ in order to maintain Bayes-plausibility of the distribution. Figure \ref{fig:ellVals} also shows the prior mean $\pi$ of the truncated normal prior. The effects of truncating the normal distribution are well-known; I only note that $\pi$ increases with $\sigma$ for $\mu < 1/2$ and decreases with $\sigma$ for $\mu > 1/2$ because $\mu = 1/2$ is the midpoint of the truncation interval.

The key value of interest is the smallest interior point of intersection $q_i(\beta(0), 0)$ between the $0$-optimal DTU and the prior distribution; I abbreviate this value to $q_i$ for the remainder of this discussion. Figure \ref{fig:qVals} illustrates some of the complex interactions between $q_i$ and the shape of the prior, as well as some of the difficulties faced in the numerical approach. For small $\mu$, $q_i$ is inverse-U-shaped as a function of $\sigma$. As discussed earlier, $\sigma$ large enough means that a lower truncation region may not be necessary. In this case I set $q_i = 0$ to preserve continuity of $q_i$ in $\sigma$ and to reflect that the Bayes-plausibility constraint does not bind at any interior intersection (it trivially binds at a posterior mean of 0, and binds at $q_i$ when a lower truncation is necessary). This case appears in panel 1. As $\mu$ increases, the curve flattens and moves up, as seen in panels 2 and 3. For $\mu > 1/2$, $q_i$ becomes monotonically increasing in $\sigma$, as shown in panel 4. In panels 5 and 6, because $F$ is convex on $[\varepsilon, 1]$ for small $\varepsilon \geq 0$, all interior intersections between the 0-optimal DTU and $F$ are either close to 0 or close to 1. The numerical algorithm thus becomes unstable and alternates between these two regions (as seen in panel 5) or chooses the default solution of $q_i = 0$ (as seen in panel 6). Nevertheless, the shape of the curve for small $\sigma$ suggests that the trend of monotonically increasing $q_i$ when $\mu > 1/2$ is preserved. 

With respect to the relationship between $q_i$ and $\pi$, which determines the gap between Propositions \ref{prop:contSmall} and \ref{prop:contLarge}, the numerical results show that the cutoff of $\mu = 1/2$ is key. For any $\mu < 1/2$, there exists $\bar{\sigma}$ large enough so that $\pi > q_i$ for all $\sigma > \bar{\sigma}$. However, $\bar{\sigma}$ may be so large that the assumption $f'(0) < 1 - 2\pi$ is violated, in which case $\pi < q_i$ for all valid choices of $\sigma$. For $\mu > 1/2$, $\pi < q_i$ for all $\sigma$; to verify this numerically, I spot-checked values $\sigma \in \left\{1, 10, 100\right\}$ and manually debugged the numerical integration. Thus there is only a gap between Propositions \ref{prop:contSmall} and \ref{prop:contLarge} when $\mu < 1/2$ and $\sigma$ is large enough to exceed $\bar{\sigma}$ but not so large as to violate $f'(0) < 1 - 2\pi$. For example, in panel 1 we can see that for $\sigma \in (0.06, 0.165)$, there is a gap between the two propositions. For $\sigma < 0.06$, $q_i > \pi$ so there is no gap, and for $\sigma > 0.165$, the assumption $f'(0) < 1 - 2\pi$ is violated and the propositions do not apply.

\end{document}